\documentclass{article}

 \usepackage[preprint]{neurips_2025}
\usepackage{conference}

\title{Thermal Min-Max Games: Unifying Bounded Rationality and Typical-Case Equilibrium}

\author{
  Yuma Ichikawa\thanks{Code: \url{https://github.com/Yuma-Ichikawa/min-max-game-replica}.} \\
  Fujitsu Limited\\
  RIKEN center for AIP \\
  \texttt{ichikawa.yuma@fujitsu.com} 
}

\begin{document}

\maketitle

\begin{abstract}
    Strategic-form min-max game theory examines the existence, multiplicity, selection of equilibria, and the worst-case computational complexity under perfect rationality. However, in many applications, games are drawn from an ensemble, and players exhibit bounded rationality. 
    We introduce thermal min-max games, a thermodynamic relaxation that unifies bounded and perfect rationality by assigning each player a temperature to regulate their rationality level. To analyze typical behavior in the large-strategy limit, we develop a nested replica framework for this relaxation. This theory provides tractable predictions for typical equilibrium values and mixed-strategy statistics as functions of rationality strength, strategy-count aspect ratio, and payoff randomness. Numerical experiments demonstrate that these asymptotic predictions accurately align with the equilibrium of finite games of moderate size.
\end{abstract}

\section{Introduction}\label{sec:introduction}
Strategic-form min-max games are a cornerstone of game theory and optimization. They formalize decision-making against an adversary or under worst-case uncertainty and are prevalent in machine learning, economics, market design, security, defense, and robust operations research whenever guaranteed performance is required against strategic opponents or disturbances \citep{Goodfellow2014GAN,Madry2018,BenTal2009RobustOpt,BasarOlsder1999}.
Classical work established min-max theorems and the foundations of equilibrium \citep{vonNeumann1928,Sion1958,BasarOlsder1999}, developing a robust theory of Nash equilibrium that encompasses existence, characterization, and stability \citep{Nash1950EquilibriumPoints,Nash1951NonCooperativeGames,Glicksberg1952KakutaniGeneralization,KohlbergMertens1986StrategicStability,Harsanyi1973RandomlyDisturbedPayoffs}.
Subsequent research focused on equilibrium selection \citep{HarsanyiSelten1988EquilibriumSelection,Schelling1960StrategyOfConflict,CarlssonVanDamme1993GlobalGames} and clarified the worst-case computational complexity of equilibrium computation \citep{LemkeHowson1964Bimatrix,DaskalakisGoldbergPapadimitriou2009ComplexityNash}.
These results provide a rigorous foundation for perfectly rational play.

However, many real-world settings deviate from perfect rationality. In learning and control, players often employ finite-time algorithms rather than omniscient optimizers: an adversary may take only a few gradient steps; responses can be noisy or heuristic; and both sides may operate under fundamentally different computational budgets. Motivated by such scenarios, a substantial body of literature has introduced bounded rationality game models and regularized response rules, including noisy best responses, entropy-regularized objectives, and quantal response equilibrium \citep{McKelveyPalfrey1995QRE,GoereeHoltPalfrey2005RegularQRE,Blume1993LogitDynamics,HofbauerSandholm2002StochasticFP}. However, these lines of research have largely evolved in parallel, leaving us without a unified account of how equilibrium values and mixed-strategy structures change as players become more or less rational, particularly when rationality is asymmetric among them.

A second challenge is that the typical-case behavior of analytical tools remains underdeveloped.
Much of game theory emphasizes existence, multiplicity, selection, and worst-case complexity. In contrast, many real-world problems are addressed as a single realization from an ensemble of instances generated by uncertainty sets, security assumptions, randomized environments, or stochastic learning dynamics.
In such distributional settings, worst-case analyses and carefully constructed counterexamples may not represent common observations.
What we often require instead are predictions for typical equilibrium outcomes, such as the standard game value and the typical sparsity or concentration of mixed strategies, as functions of payoff randomness and the strength of rationality.

This study develops a typical-case theory that bridges both gaps. We introduce thermal min-max games, a thermodynamic relaxation of strategic-form min-max games in which each player possesses an inverse-temperature parameter that regulates the degree of rationality. At finite temperatures, play is governed by soft best responses.
In sequential leader-follower variants \citep{McKelveyPalfrey1998AQRE}, the temperature quantifies how precisely the follower conditions on the leader's action, given limited computational precision. In simultaneous-move variants \citep{McKelveyPalfrey1995QRE,GoereeHoltPalfrey2005RegularQRE,Blume1993LogitDynamics,HofbauerSandholm2002StochasticFP}, the prior distribution in thermal games serves as a regularizer on responses, modeling decision noise and uncertainty. In the ordered zero-temperature limit, thermal min-max games converge to the classical Nash value, thus providing an interpolation between bounded-rationality and perfectly rational regimes.

To characterize the typical equilibrium under payoff randomness, we develop a nested replica framework tailored to the two-level structure of thermal games. Our analysis of random payoff ensembles leads to tractable characterizations of the typical equilibrium value and mixed--strategy statistics, including concentration and sparsity, explicitly expressed in terms of the two temperatures, the strategy-count aspect ratio between players, and the underlying payoff distribution. Numerical experiments show that these asymptotic predictions closely match equilibrium in finite games of moderate size, indicating that the theory remains accurate well beyond the strict thermodynamic limit.

\paragraph{Contributions.}
We summarize our main contributions as follows:
\begin{itemize}
  \item \textbf{A unified two-temperature model of bounded rationality.}
  We propose thermal min-max games that incorporate player-specific rationality and include several standard bounded rationality models as special cases.
 Taking the ordered zero-temperature limit, the model reduces to the standard perfectly rational min-max game.
  \item \textbf{Nested replica theory for typical-case equilibrium.}
  We develop a nested replica framework that provides asymptotically exact predictions for typical-cases in the large-strategy limit.
  For Gaussian payoff ensembles, we derive explicit formulas for the typical equilibrium value and the statistics of mixed strategies.
\end{itemize}

\section{Problem Setting}
\paragraph{Strategic-Form Min-Max Game}
We start with a standard finite zero-sum matrix game.
Let $\B{C}\in\mab{R}^{N\times M}$ be an arbitrary payoff matrix.
The minimizer selects a mixed strategy $\B{p}\in\Delta_N$ with
$\Delta_N=\{\B{p}\in\mab{R}^N:\ p_i\ge 0,\ \B{p}^{\top}\1=1\}$,
while the maximizer chooses $\B{q}\in\Delta_M$ with
$\Delta_M=\{\B{q}\in\mab{R}^M:\ q_j\ge 0,\ \B{q}^{\top}\1=1\}$.
The payoff is represented by the bilinear form:
\begin{equation}
    e(\B{p},\B{q};\B{C})=\B{p}^{\top}\B{C}\B{q}.
    \label{eq:zero-sum-payoff-unscaled}
\end{equation}
The Nash value of the game is the min-max value $t(\B{C})=\min_{\B{p}}\max_{\B{q}}e(\B{p},\B{q};\B{C})$
, which is equal to the max-min value according to von Neumann's min-max theorem \citep{vonNeumann1928}.
A pair $(\B{p}^{\star},\B{q}^{\star})$ is a Nash equilibrium if it achieves this value, meaning that
$e(\B{p}^{\star},\B{q};\B{C})\le t(\B{C})\le e(\B{p},\B{q}^{\star};\B{C})$ holds for all
$\B{p}\in\Delta_N$ and $\B{q}\in\Delta_M$.

\paragraph{Typical-Case Analysis.}
Research on strategic-form min-max games has traditionally focused on three core questions: the existence of Nash equilibrium, the origin and resolution of equilibrium multiplicity, and the difficulty of computing an equilibrium in the worst case. A typical-case viewpoint treats the payoff matrix $\B{C}$ as a random draw from an instance ensemble and investigates the resulting distribution of the game value $t(\B{C})$. 
Specifically, it characterizes the mean $\mab{E}_{\B{C}\sim p(\B{C})}[t(\B{C})]$ and the fluctuations between instances.
Statistical mechanics provides a framework for typical-case analysis in large degree-of-freedom limits  \citep{mezard1986replica,fontanari1995statistical,IchikawaHukushima2025TMLR}. Motivated by this perspective, we examine strategic-form games in an asymptotic regime where the number of available strategies increases accordingly.

\begin{definition}[Thermodynamic Limit]
    \label{def:thermodynamic-limit}
    The thermodynamic limit for an $N\times M$ instance is the asymptotic regime
    \begin{equation}
    N,M\to\infty,~~
    \nicefrac{N}{M}=\gamma\in(0,\infty).
    \label{eq:thermodynamic-limit}
    \end{equation}
\end{definition}

\paragraph{Scaled Min-Max Game.}
To take the thermodynamic limit in Definition~\ref{def:thermodynamic-limit}, we utilize rescaled simplices
\begin{equation}
    \mac{X}_N=\ab\{\B{x}\in\mab{R}^N: x_i\ge 0,~ \B{x}^{\top} \1 = N\},~~
    \mac{Y}_M=\ab\{\B{y}\in\mab{R}^M: y_j\ge 0,~ \B{y}^{\top}\1=M \},
    \label{eq:scaled-simplices}
\end{equation}
so that typical coordinates are $\mac{O}(1)$. These variables correspond to standard mixed strategies through $p_i=\nicefrac{x_i}{N}$ and $q_j=\nicefrac{y_j}{M}$.
Given $\B{C} \sim p(\B{C})$, we define the bilinear payoff
\begin{equation}
    E(\B{x},\B{y};\B{C})
    =
    \kappa \ab(\B{x}^{\top} \B{C} \B{y}),
    \label{eq:general-scaled-payoff}
\end{equation}
where the normalization $\kappa$ is chosen such that the equilibrium value is extensive in the thermodynamic limit. Under mild moment assumptions on $p(\B{C})$, 
this scaling yields a min-max value that is typically $\Theta(L)$ with $L=\sqrt{NM}$. The intensive quantity $\nicefrac{E(\B{x}, \B{y}; \B{C})}{L}$ can have a nontrivial limit as $N, M \to \infty$.

\section{Thermal Min-Max Games}
\label{sec:two-temp-br-and-typical}

We introduce \emph{thermal min-max games} as a model of bounded rationality, achieved by replacing hard best responses with Boltzmann sampling at finite inverse temperatures $(\beta_{\max},\beta_{\min})\in(0,\infty)^2$.

\begin{definition}[Thermal Min-Max Games]
\label{def:ttbs}
Fix $(\beta_{\max},\beta_{\min})\in(0,\infty)^2$ and an instance $\B{C}$.
A Two-Temperature Boltzmann Strategy (TTBS) consists of a random pair of strategies $(\B{X},\B{Y})$ generated by
\begin{equation}
    \B{X}\sim p_{\beta_{\min}}(\cdot;\B{C}),~~~
    \B{Y}|  \B{X}\sim p_{\beta_{\max}}(\cdot|\B{X};\B{C}),
    \label{eq:ttbs-sampling}
\end{equation}
where
\begin{align}
    p_{\beta_{\max}}(\B{y}| \B{x};\B{C})
    &=
    \frac{\exp\ab(\beta_{\max}E(\B{x},\B{y};\B{C}))}{Z_y(\B{x};\B{C})}q_{0}(\B{y}),
    Z_y(\B{x};\B{C})
    =
    \int_{\mac{Y}_M}q_{0}(\B{y})\exp\ab(\beta_{\max}E(\B{x},\B{y};\B{C}))d\B{y},
    \label{eq:ttbs-inner-def}
    \\
    p_{\beta_{\min}}(\B{x};\B{C})
    &=
    \frac{\exp\ab(-\beta_{\min}\phi_{\beta_{\max}}(\B{x}))}{Z(\B{C})}p_{0}(\B{x}),~
    Z(\B{C})
    =
    \int_{\mac{X}_N} p_{0}(\B{x})
    \exp\ab(-\beta_{\min}\phi_{\beta_{\max}}(\B{x};\B{C}))d\B{x}.
    \label{eq:ttbs-outer-def}
\end{align}
with $\phi_{\beta_{\max}}(\B{x})\coloneqq (\nicefrac{1}{\beta_{\max}})\log Z_y(\B{x};\B{C})$.
Here $q_{0}(\B{y})$ and $p_{0}(\B{x})$ denote the prior distributions over $\mac{Y}_{M}$ and $\mac{X}_{N}$, respectively.
We define the averaged expected payoff and the corresponding two-temperature free energy as
\begin{equation}
    \ab\langle E(\B{X},\B{Y};\B{C}) \rangle_{\beta_{\max},\beta_{\min}} \coloneqq \mab{E}_{\B{X}, \B{Y}}\ab[E(\B{X},\B{Y};\B{C})],~~~F(\beta_{\max},\beta_{\min};\B{C})
    \coloneqq
    -\frac{1}{\beta_{\min}}\log Z(\B{C}).
\end{equation}
The free-energy value $F$ serves as a log-partition normalization and incorporates the entropic contributions induced by TTBS.
\end{definition}

In statistical physics, $\beta=\nicefrac{1}{T}$ controls concentration: high temperatures produce diffuse sampling, whereas low temperatures focus on near-optimal configurations.
In our setting, $\beta_{\max}$ and $\beta_{\min}$ control the strength of rationality: smaller values correspond to weaker rationality, whereas larger values correspond to stronger rationality.
In particular, the ordered zero-temperature limit first $\beta_{\max}\to\infty$ and then $\beta_{\min}\to\infty$ recovers the classical Nash value $E_{0}(\B{C})$.

\paragraph{A Unified View of Bounded Rationality via Thermal Min-Max Games.}
These thermal min-max game offers a unified perspective that integrates several established min-max formulations. By selecting appropriate priors, our framework captures bounded-rational relaxations of simultaneous-move play. In contrast, the two temperatures $(\beta_{\max},\beta_{\min})$ parameterize an inherently asymmetric leader-follower relaxation induced by the nested construction.
Specifically, under the ordered zero-temperature limit, with an entropic specification
$p_{0}(\B{x}) \propto \exp(-\beta_{\min}\lambda_{\min} H(\B{x}))$ and
$q_{0}(\B{y}) \propto \exp(\beta_{\max}\lambda_{\max} H(\B{y}))$,
where $H$ denotes an entropy function, the induced equilibrium coincides with the quantal response equilibrium (QRE) \citep{McKelveyPalfrey1995QRE,GoereeHoltPalfrey2005RegularQRE,Blume1993LogitDynamics,HofbauerSandholm2002StochasticFP}.
Our formulation also recovers the logit agent-quantal response equilibrium (AQRE) \citep{McKelveyPalfrey1998AQRE} as a special case.
We refer to Appendix~\ref{app:bounded_rationality} for detailed derivations of these reductions and a broader discussion of connections to other bounded rationality models.

\section{Nested Replica Framework for Thermal Min-Max Games}
\label{sec:replica-route-extended}

Typical-case analysis addresses the average behavior across an ensemble of random instances.
For a statistic $\Gamma:\mac{X}_N\times\mac{Y}_M\to\mab{R}$, we examine its instance-averaged value under TTBS sampling, $
\mab{E}_{\B{C}}[\langle \Gamma(\B{x},\B{y})\rangle_{\beta_{\max},\beta_{\min}}]$.
Representative observables include the payoff $\Gamma(\B{x},\B{y})=E(\B{x},\B{y};\B{C})$, coordinatewise moments of mixed strategies, and summary statistics that describe the geometry of equilibrium.
Importantly, such TTBS observables can be derived from a generating function: the typical free-energy density, defined as follows:
\begin{equation}
    \nu(\beta_{\max},\beta_{\min})
    =
    \lim_{L\to\infty}\frac{1}{L}\mab{E}_{\B{C}}\ab[F(\beta_{\max},\beta_{\min};\B{C})].
    \label{eq:typical-v-general}
\end{equation}
In practice, derivatives of $\nu$ with respect to model parameters such as $\gamma$ and $k$ generate a broad class of typical double-expectation statistics.
In general, the instance average $\mab{E}_{\B{C}}[\log Z(\B{C})]$ is not analytically tractable. However, in the thermodynamic limit of Definition~\ref{def:thermodynamic-limit}, it can be evaluated using the replica method \citep{edwards1975theory, mezard1986replica, mezard2009information}. The central idea is to express $\nu(\beta_{\max},\beta_{\min})$ in terms of moments using the replica identity
\begin{equation}
    \nu(\beta_{\max},\beta_{\min})
    =
    -\frac{1}{\beta_{\min}}
    \lim_{n\to 0}\frac{1}{n}\lim_{L\to\infty}\frac{1}{L}\log \mab{E}_{\B{C}}\ab[Z(\B{C})^n],
    \label{eq:v-from-replica}
\end{equation}
The replica method computes $\mab{E}_{\B{C}}[Z(\B{C})^{n}]$ for integers $n$ and $k$, and then analytically continues the resulting expression to the limits $n\to 0$ and $k\to -\nicefrac{\beta_{\min}}{\beta_{\max}}$. For integers $n$ and $k$,
\begin{equation}
    Z(\B{C})^n =
    \int \prod_{a=1}^n d\B{x}^a
    \prod_{a=1}^n\prod_{l=1}^k d\B{y}^{al}
    \exp\ab(\beta_{\max}\sum_{a=1}^n\sum_{l=1}^k E(\B{x}^a,\B{y}^{al};\B{C})),
    \label{eq:replicated-Z}
\end{equation}
Eq.~\eqref{eq:replicated-Z} reveals a nested replica structure: $a\in\{1,\dots,n\}$ is the standard replica index, while $l\in\{1,\dots,k\}$ denotes a temperature replica induced by the power $[Z_y(\B{x};\B{C})]^k$.
Establishing these analytic continuations rigorously can be challenging.
All results in this paper are derived using the nested replica method and are interpreted as conditional on the validity of the analytic continuation steps.
A growing body of literature has validated replica predictions for several high-dimensional inference models \citep{barbier2019optimal,aubin2020generalization}; however, a formal treatment of the current two-temperature min-max construction is beyond the scope of this work and will be addressed in future research.

For the Gaussian bilinear ensemble, averaging Eq.~\eqref{eq:replicated-Z} over $\B{C}$ induces quadratic couplings among replicas.
Importantly, these couplings depend on the replicated configurations through overlap matrices such as
$Q^x_{ab}=\nicefrac{\B{x}^{a \top} \B{x}^{b}}{N}$ and $Q^y_{abls}=\nicefrac{\B{y}^{al \top} \B{y}^{bs}}{M}$.
To perform the analytic continuation and solve the resulting saddle-point equations, we impose a structured \emph{ansatz} on these overlaps:
we assume replica symmetry (RS) for the $n$ outer replicas of $\B{x}$ and a one-step replica symmetry-breaking (1RSB) structure for the inner replicas of $\B{y}$, as illustrated in Figure~\ref{fig:fig_rs_1rsb_overlap_only}.

\begin{figure}
    \centering
    \includegraphics[width=0.95\linewidth]{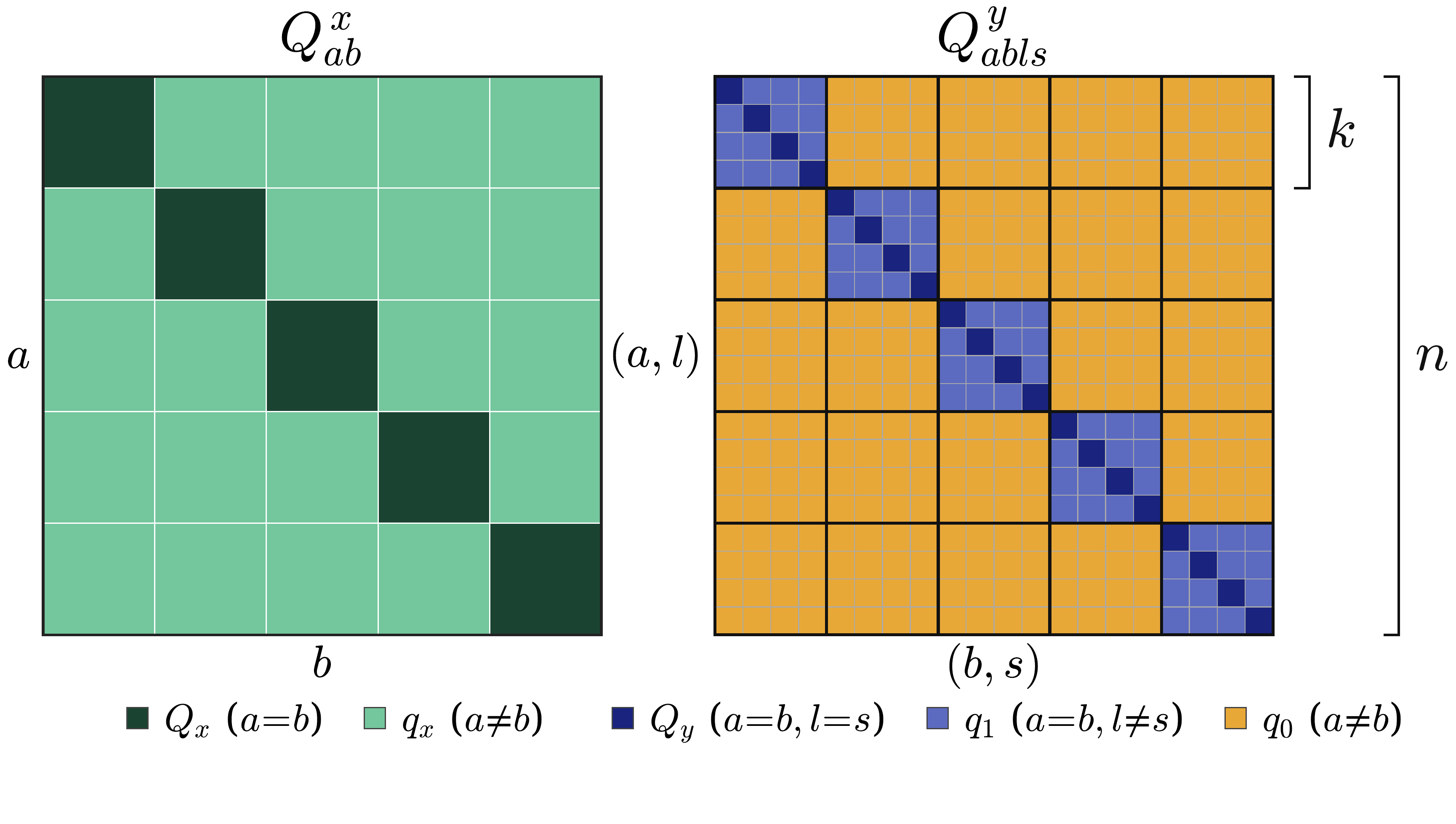}
    \caption{Replica-overlap structure under the RS and 1RSB ansatz. Left: RS overlap matrix $Q^x_{ab}$ for the $n$ outer replicas of $\B{x}$, with diagonal entries $Q_x$ ($a=b$) and off-diagonal entries $q_x$ ($a\neq b$). Right: 1RSB overlap matrix $Q^y_{abls}$ for the $\B{y}$ replicas indexed by $(a,l)$, showing $k$ inner replicas within each outer replica $a$: $Q_y$ on $(a=b,l=s)$, $q_1$ on $(a=b,l\neq s)$, and $q_0$ on $(a\neq b)$.}
    \label{fig:fig_rs_1rsb_overlap_only}
\end{figure}

\section{Typical Behavior in Gaussian Thermal Games}
\label{sec:main-results}

We present results for the typical value density $\nu(\beta_{\max},\beta_{\min})$ in a Gaussian setting. We specifically consider payoff matrices with independent entries $C_{ij}\sim_{\mathrm{i.i.d.}}\mac{N}(0,1)$ and adopt the scaling $\kappa=\sqrt{\nicefrac{\sigma}{L}}$, under which the min-max value satisfies $E_0(\B{C})=\Theta(L)$ in the thermodynamic regime.
We further specialize in uniform priors $p_{0}(\B{x})$ and $q_{0}(\B{y})$.
Although we focus on i.i.d. Gaussian payoffs for clarity, the two-temperature framework is not limited to this assumption.
Extending the typical-case theory to structured ensembles, including low-rank-plus-noise models, sparse interactions, and payoffs with correlated rows or columns, is a natural and promising direction for future work.
All technical derivations are provided in Appendix~\ref{app-sec:derivation}.

\begin{claim}
    \label{thm:finiteT}
    Fix $\sigma>0$, the aspect ratio $\gamma>0$, and the inverse temperatures $\beta_{\max}>0$ and $\beta_{\min}>0$.
    Let $k=-\nicefrac{\beta_{\min}}{\beta_{\max}}<0$, and consider the thermodynamic limit of Definition~\ref{def:thermodynamic-limit}.
    Assume the RS and 1RSB ansatz as described in Appendix~\ref{app-subsec:rs-rsb}.
    Then, the typical value density has a saddle-point representation
    \begin{equation}
        \nu(\beta_{\max},\beta_{\min})
        =
        -\frac{1}{\beta_{\min}} \mathrm{extr}_{\B{\Theta}}g(\Theta; \B{\Lambda}),
    \label{eq:v-extr-g}
    \end{equation}
    where $\mathrm{extr}_{\B{\Theta}}$ denotes a stationary point over the finite-dimensional collection of order parameters and conjugates
    \begin{equation}
        \B{\Theta} \coloneqq
        \ab(
        Q_x,q_x,Q_y,q_0,q_1,
        m_x,m_y,
        \hat{Q}_x,\hat{\chi}_x,\hat{Q}_y,\hat{\chi}_0,\hat{\chi}_1),
        \label{eq:Theta-def}
    \end{equation}
    and $\B{\Lambda}$ denotes the collection of scalar model parameters such as $\sigma$, $\gamma$, $\beta_{\max}$, and $\beta_{\min}$.
    The function $g(\B{\Theta}; \B{\Lambda})$ is given by
    \begin{multline}
        g(\B{\Theta}; \B{\Lambda})
        =
        \frac{\sigma\beta_{\max}^2}{2}\ab(
        kQ_xQ_y+k(k-1)Q_x q_1-k^2 q_x q_0)
        +
        \frac{\gamma^{\nicefrac{1}{2}}}{2}\ab(\hat{Q}_x Q_x-\hat{\chi}_x(Q_x-q_x))\\
        +
        \frac{\gamma^{-\nicefrac{1}{2}}}{2}\ab(
        k\hat{Q}_y Q_y
        -k(\hat{\chi}_1+\hat{\chi}_0)\bigl(Q_y+(k-1)q_1\bigr)
        +k^2\hat{\chi}_0 q_0)
        -\gamma^{\nicefrac{1}{2}}m_x-k\gamma^{-\nicefrac{1}{2}}m_y\\
        +\gamma^{\nicefrac{1}{2}}\int Dz \log Z_x(z)
        + \gamma^{-\nicefrac{1}{2}} \int Dz \log\int D\eta [Z_y(z,\eta)]^k,
        \label{eq:g-functional}
    \end{multline}
    where $Dz=(2\pi)^{-\nicefrac{1}{2}}e^{-\nicefrac{z^2}{2}}dz$ and $D\eta=(2\pi)^{-\nicefrac{1}{2}}e^{-\nicefrac{\eta^2}{2}}d\eta$ denote standard Gaussian measures, and
    \begin{multline}
        Z_x(z)=\int_{0}^{N}d x  \exp\ab(-\frac{\hat{Q}_x}{2}x^2+\ab(m_x+\sqrt{\hat{\chi}_x}z)x),\\
        Z_y(z,\eta)=\int_{0}^{M}d y  \exp\ab(-\frac{\hat{Q}_y}{2}y^2+\ab(m_y+\sqrt{\hat{\chi}_0} z+\sqrt{\hat{\chi}_1}\eta )y).
    \label{eq:onesite-Zxy}
    \end{multline}
    Since the upper boundaries are asymptotically inactive, i.e., the saddle-point values satisfy $x^\star<N$ and $y^\star<M$, we replace the one-site domains $[0,N]$ and $[0,M]$ by $[0,\infty)$ in the thermodynamic limit.
\end{claim}

\begin{proposition}
    \label{prop:finiteT-stationarity}
     At a saddle point $\B{\Theta}^\star$ of $g(\B{\Theta}; \B{\Lambda})$, the conjugate variables satisfy
     \begin{multline}
         \hat{Q}_x=\frac{k\sigma\beta_{\max}^{2}}{\gamma^{\nicefrac{1}{2}}}\ab(k (q_{0}-q_{1})-(Q_{y}-q_{1})),~~\hat{Q}_y=0,\\        \hat{\chi}_x=\frac{\sigma\beta_{\max}^2}{\gamma^{\nicefrac{1}{2}}}k^2 q_0,~~~
        \hat{\chi}_0=\gamma^{\nicefrac{1}{2}}\sigma\beta_{\max}^2q_x,~~~
        \hat{\chi}_1=\gamma^{\nicefrac{1}{2}}\sigma\beta_{\max}^2(Q_x-q_x),
        \label{eq:finiteT-conjugates}
     \end{multline}
    and the order parameters obey the moment self-consistency relations
    \begin{multline}
        \int Dz \langle x\rangle_{z}=1,~~
        \int Dz \ab\langle \langle y\rangle_{z,\eta} \rangle_{\eta|z}=1,~~
        Q_x=\int Dz \langle x^2\rangle_{z},~~
        q_x=\int Dz \langle x\rangle_{z}^2,\\
        Q_y=\int Dz \ab\langle \langle y^2\rangle_{z,\eta}\rangle_{\eta|z},~~
        q_1=\int Dz \ab\langle \langle y\rangle_{z,\eta}^2 \rangle_{\eta|z},~~
        q_0=\int Dz \ab(\langle \langle y\rangle_{z,\eta} \rangle_{\eta|z})^2.
    \end{multline}
    where $\langle\cdot\rangle_{z}$ and $\langle\cdot\rangle_{z,\eta}$ denote the Gibbs expectations under the one-site measures induced by $Z_x(z)$ and $Z_y(z,\eta)$, with $\langle\cdot\rangle_{\eta|z}$ defined as follows: 
    \begin{equation}
        \Psi(z) \coloneqq \int D\eta[Z_y(z,\eta)]^k,~~
        \langle B \rangle_{\eta|z}
        \coloneqq
        \frac{1}{\Psi(z)}\int D\eta[Z_y(z,\eta)]^k B(z,\eta),
        \label{eq:eta-tilt}
    \end{equation}
    for any integrable $B(z,\eta)$.
\end{proposition}

\begin{proposition}
\label{prop:finiteT-payoff}
    The typical TTBS payoff density is characterized by saddle parameters as
    \begin{equation}
    e(\beta_{\max},\beta_{\min})
    \coloneqq 
    \lim_{L\to\infty}\frac{1}{L}
    \mab{E}_{\B{C}}\ab[\langle E(\B{X},\B{Y};\B{C})\rangle_{\beta_{\max},\beta_{\min}}]
    =
    \sigma\beta_{\max}\Bigl(
    Q_xQ_y+(k-1)Q_xq_1-kq_xq_0
    \Bigr).
\label{eq:finiteT-payoff-density}
\end{equation}
\end{proposition}
Claim \ref{thm:finiteT} shows that, in the thermodynamic limit, the nested two-stage relaxation is characterized by a finite set of scalar order parameters $\B{\Theta}$.
These order parameters summarize the geometry of equilibrium. For instance,
$Q_x=\lim_{L\to\infty}\mab{E}_{\B{C}}[\langle \nicefrac{\|\B{X}\|^2\rangle_{\beta_{\max},\beta_{\min}}}{N}]$
and
$Q_y=\lim_{L\to\infty}\mab{E}_{\B{C}}[\langle \nicefrac{\|\B{Y}\|^2\rangle_{\beta_{\max},\beta_{\min}}}{M}]$
represent the typical self-overlaps of mixed strategies under the two-temperature Gibbs measure.
More broadly, once the saddle-point solution $\B{\Theta}^{\ast}$ is obtained, other typical observables can be computed from low-dimensional Gaussian integrals over the effective fields $(z,\eta)$, i.e., as expectations of the corresponding single-site moments.

\subsection{Perfect Rationality: Ordered Zero-Temperature Limit}
\label{sec:ordered-zeroT-main}

To study the min-max game under perfect rationality, we take the ordered zero-temperature limit of the finite-temperature saddle-point equations. 

\begin{claim}
    \label{claim:oztl}
    Consider the thermodynamic limit, followed by the ordered zero-temperature limit.
    Assume the RS and 1RSB ansatz.
    Let $\Phi(a)=\int_{-\infty}^{a}(2\pi)^{-\nicefrac{1}{2}}e^{-\nicefrac{z^{2}}{2}}dz$ denote the standard normal CDF and $\varphi(a)=(2\pi)^{-\nicefrac{1}{2}}e^{-\nicefrac{a^{2}}{2}}$ the corresponding PDF.
    Define the truncated-Gaussian functions
    \begin{equation}
    q(a)=\nicefrac{B(a)}{A(a)^2},~~A(a)=\mab{E}_{Z\sim\mac{N}(0,1)}[(Z+a)_{+}],~~
    B(a)=\mab{E}_{Z\sim\mac{N}(0,1)}[(Z+a)_{+}^2]
    \label{eq:AB-def}
    \end{equation}
    where $(u)_{+} \coloneqq \max\{u,0\}$.
    The saddle is characterized by $(\alpha_x,\alpha_y)\in\mab{R}^{2}$, satisfying
    \begin{align}
        \Phi(\alpha_y)=\gamma \Phi(\alpha_x),~~~\gamma^{\nicefrac{1}{2}}\alpha_x q^{\nicefrac{1}{2}}(\alpha_y)+\alpha_y q^{\nicefrac{1}{2}}(\alpha_x)=0.
    \end{align}
    The average Nash value density is
    \begin{equation}
        \frac{E_0}{L}= \frac{\sigma^{\nicefrac{1}{2}}}{2}\ab(\gamma^{-\nicefrac{1}{4}}\alpha_x q^{\nicefrac{1}{2}}(\alpha_y)-\gamma^{\nicefrac{1}{4}}\alpha_y q^{\nicefrac{1}{2}}(\alpha_x)).
        \label{eq:f-gamma}
    \end{equation}
    Moreover, the equilibrium strategies adhere to the one-site laws
    \begin{equation}
        Z_x,Z_y \stackrel{\text{i.i.d.}}{\sim}\mac{N}(0,1),~~~
        X=\frac{(Z_x+\alpha_x)_{+}}{A(\alpha_x)},~~~
        Y=\frac{(Z_y+\alpha_y)_{+}}{A(\alpha_y)}.
    \label{eq:onesite-law}
    \end{equation}
    In particular, for uniformly random coordinates $i\in\{1,\dots,N\}$ and $j\in\{1,\dots,M\}$, the empirical distributions of the equilibrium components $x_i^\star$ and $y_j^\star$ converge to the laws of $X$ and $Y$, respectively.
\end{claim}

\begin{proposition}
\label{prop:oztl-stats}
The asymptotic support fractions and second moments are provided by
\begin{multline}
    \rho_x \coloneqq \lim_{N\to\infty}\frac{1}{N} \mab{E}_{\B{C}}\ab[\#\{i:\ x_i^\star>0\}]=\Phi(\alpha_{x}),~~~
    \rho_y \coloneqq \lim_{M\to\infty}\frac{1}{M}\mab{E}_{\B{C}}\ab[\#\{j:\ y_j^\star>0\}]=\Phi(\alpha_{y}),\\
    q_x \coloneqq \lim_{N\to\infty}\frac{1}{N}\mab{E}_{\B{C}}\ab[\|\B{x}^{\star}\|^{2}]=q(\alpha_{x}),~~~
    q_y \coloneqq \lim_{M\to\infty}\frac{1}{M}\mab{E}_{\B{C}}[\|\B{y}^\star\|^2]=q(\alpha_{y}).
    \label{eq:rho-def}
\end{multline}
\end{proposition}

The support fractions $(\rho_x,\rho_y)$ quantify how many pure strategies receive positive weight at equilibrium; they are the limiting fractions of active coordinates in the scaled mixed strategies $(\B{x}^\star,\B{y}^\star)$.
The second moments $(q_x,q_y)$ capture the concentration and heterogeneity of the equilibrium weights; larger values of $q_x$ and $q_y$ indicate that the minimizer and maximizer allocate probability mass more unevenly over their active supports.

\section{Perturbative Analysis and Linear Response}
\label{sec:perturbative-summaries-main}

Building on the above claims, we conduct perturbative expansions in both the payoff randomness and the strategy-number aspect ratio $\gamma$ to characterize typical behavior.
The saddle-point representations in Claims~\ref{thm:finiteT} and~\ref{claim:oztl} provide a convenient calculus: once the macroscopic order parameters satisfy the stationarity conditions, we can systematically derive perturbation series and linear-response relationships concerning the model parameters. Specifically, at a stationary point $\B{\Theta}^\star(\B{\Lambda})$ satisfying $\partial_{\B{\Theta}} g(\B{\Theta}^\star(\B{\Lambda});\B{\Lambda})=0$, we obtain
\begin{equation}
    \partial_{\B{\Lambda}} \nu(\beta_{\max},\beta_{\min})
    =
    -\left.\beta_{\min}^{-1}\partial_{\B{\Lambda}} g(\B{\Theta};\B{\Lambda})\right|_{\B{\Theta}=\B{\Theta}^\star(\B{\Lambda})}.
    \label{eq:envelope-main}
\end{equation}
Eq.~\eqref{eq:envelope-main} does not require the differentiation of $\B{\Theta}^\star(\B{\Lambda})$.

\subsection{Perturbation in Payoff Randomness}
\label{sec:main-sigma-expansion}

We investigate how typical behavior changes as a function of the payoff-randomness strength $\sigma$. This dependence is crucial: in many applications, $\sigma$ signifies noise, heterogeneity, or uncertainty in realized payoffs. Analyzing the small-$\sigma$ regime provides a controlled baseline, allowing us to quantify when and how bounded-rational strategies diverge from nearly uniform play.

\begin{proposition}
\label{prop:sigma-expansion-main}
Fix $\gamma>0$ and $\beta_{\max},\beta_{\min}>0$.
In the small-$\sigma$ regime, the typical finite-temperature free-energy density admits the expansion
\begin{equation}
    \nu(\beta_{\max},\beta_{\min};\sigma,\gamma)
    = v_{\mathrm{ent}}(\beta_{\max},\beta_{\min};\gamma) +
    \sigma\ab(\beta_{\max}-\frac{\beta_{\min}}{2})
    + \sigma^2 v_2(\beta_{\max},\beta_{\min};\gamma)
    + \mac{O}(\sigma^3),
    \label{eq:v-sigma2-main}
\end{equation}
where $v_{\mathrm{ent}}(\beta_{\max},\beta_{\min};\gamma)
= - \gamma^{\nicefrac{1}{2}} \beta_{\min}^{-1} + \gamma^{-\nicefrac{1}{2}} \beta_{\max}^{-1}$, and the coefficient $v_2(\beta_{\max},\beta_{\min};\gamma)$ is provided in closed form in Appendix~\ref{app:sigma-expansion}.
Moreover, the corresponding typical energy density satisfies
\begin{equation}
    e(\beta_{\max},\beta_{\min};\sigma,\gamma)
    = (2\beta_{\max}-\beta_{\min})\sigma
    + 4 \sigma^2 v_2(\beta_{\max},\beta_{\min};\gamma) + \mac{O}(\sigma^3).
\label{eq:e-sigma2-main}
\end{equation}
\end{proposition}
A detailed derivation is provided in Appendix~\ref{app:sigma-expansion}.
The observed asymmetry directly results from the nested TTBS construction: $\beta_{\max}$ sharpens the inner soft-max over $\B{Y}$, amplifying payoff peaks at fixed $\B{X}$, while $\beta_{\min}$ sharpens the outer soft-min over $\B{X}$, penalizing large continuation values. These distinct roles yield different prefactors in the expansion.
In the weak-disorder regime, the $\mac{O}(\sigma)$ term is independent of $\gamma$, as both players remain close to the uniform $\sigma=0$ strategies; dependence on $\gamma$ arises only at $\mac{O}(\sigma^2)$, when the strategies begin to deviate significantly from uniformity.
Accordingly, the linear coefficient $\beta_{\max}-(\nicefrac{\beta_{\min}}{2})$ quantifies an asymmetric rationality effect: increasing $\beta_{\max}$ raises the typical value, while increasing $\beta_{\min}$ reduces it.

\subsection{Perturbation around Equal-Strategy Regime}
\label{sec:main-gamma-expansion}

We examine how the Nash value varies with the strategy-count aspect ratio $\gamma$ near the balanced point.

\begin{proposition}
    \label{prop:gamma-expansion-main}
    In the ordered zero-temperature limit, let $\gamma = 1+\varepsilon$ equal $|\varepsilon|\ll 1$. The order parameters and Nash value density can be expressed as
    \begin{equation}
        \alpha_x \approx -\sqrt{\frac{\pi}{8}} \varepsilon + \frac{\sqrt{2\pi}}{16}(5-\pi)\varepsilon^2,~~\alpha_y \approx +\sqrt{\frac{\pi}{8}}\varepsilon + \frac{\sqrt{2\pi}}{16}(1-\pi)\varepsilon^2~~,\frac{E_0}{L}
        = -\frac{\pi}{2} \sqrt{\frac{\sigma}{2}} \varepsilon + \mac{O}(\varepsilon^2).
        \label{eq:E0-linear-main}
    \end{equation}
\end{proposition}
The derivation is provided in Appendix~\ref{app:gamma-expansion}.
At the balanced point $\gamma=1$, the two players exhibit comparable degrees of freedom, resulting in instances that exhibit no systematic advantage, with the value converging near zero.
A small perturbation away from $\gamma=1$ shifts the equilibrium support: the player with the higher-dimensional strategy space gains additional flexibility, yielding an advantage that increases linearly with the dimensional imbalance.

\section{Numerical Comparison with Finite Strategy Games}
\label{sec:numerical-experiments}

We compare the ordered zero-temperature RS predictions in Claim \ref{claim:oztl} for the standard Nash equilibrium with those for finite games.
On the theoretical side, we numerically solve the two scalar RS equations in Claim \ref{claim:oztl} and utilize the resulting parameters to evaluate the thermodynamic prediction Eq.~\eqref{eq:f-gamma} for the normalized game value.
On the empirical side, for each sampled payoff matrix, we compute the exact min-max value and the corresponding equilibrium strategies by solving the primal-dual linear program detailed in Appendix~\ref{app:numerics}.
We also provide a finite-temperature comparison in Appendix~\ref{sec:additional-comparison}, where we find close agreement with the replica prediction across temperatures and aspect ratios.
We fix $M=200$ and $\sigma=1$ while varying the aspect ratio $\gamma$ by setting $N=\gamma M$. We report the normalized value $v=\nicefrac{E_{0}}{L}$, the support fractions $\rho_x$ and $\rho_y$, and the second moments
$q_x=\mab{E}_{\B{C}}[\nicefrac{\|\B{x}\|^{2}}{N}]$ and $q_y=\mab{E}_{\B{C}}[\nicefrac{\|\B{y}\|^{2}}{M}]$ of the rescaled equilibrium strategies.
Figure~\ref{fig:zeroT_three_panels} demonstrates that the thermodynamic RS curves closely track the disorder-averaged LP measurements for all three
observables.
Although derived under the $N,M\to\infty$ limit, the theory yields quantitatively accurate predictions at $M=200$.
The value $\nu(\gamma)$ approaches zero as $\gamma=1$ and changes sign as the relative dimensionality flips, consistent with the intuition that a player with more degrees of freedom can typically secure a more favorable outcome.
The predicted support-size matching relation $\rho_y \approx \gamma \rho_x$ is also evident in the finite-size solutions, indicating that typical equilibrium are sparse and that the effective number of active pure strategies systematically adapts to the aspect ratio.

\begin{figure}
    \centering
    \includegraphics[width=\linewidth]{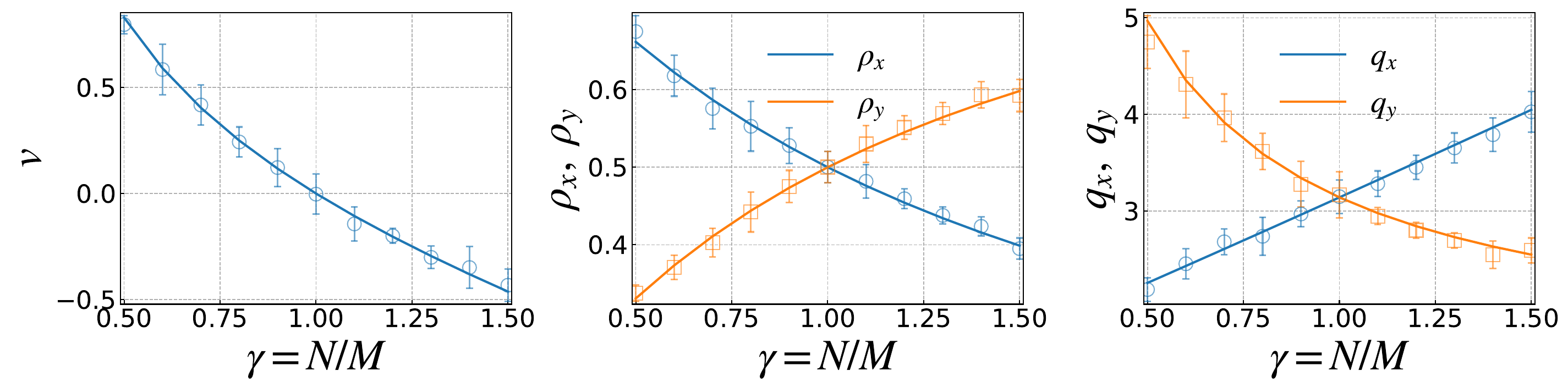}
    \caption{Comparison of theoretical predictions with finite-size LP equilibrium at $M=200$.
    Left: normalized Nash value $\nu=\nicefrac{E_0}{L}$ as a function of the aspect ratio $\gamma=\nicefrac{N}{M}$.
    Middle: support fractions $\rho_x$ and $\rho_y$.
    Right: second moments $q_x=\mab{E}_{\B{C}}[\nicefrac{\|\B{x}\|^{2}}{N}]$ and $q_y=\mab{E}_{\B{C}}[\nicefrac{\|\B{y}\|^{2}}{M}]$ of the rescaled equilibrium strategies.
    Markers denote disorder-averaged measurements obtained from the LP solver, with error bars showing standard errors over 10 random seeds; solid curves show the corresponding RS predictions from Claim~\ref{claim:oztl}.}
    \label{fig:zeroT_three_panels}
\end{figure}

\section{Related Work}
\label{sec:related-work}

\paragraph{Relaxing Perfect Rationality in Min-Max Games.}
A broad class of bounded rationality models replaces hard best responses with soft response rules. In finite games with discrete action sets, this perspective yields log-sum-exp smoothings and logit choice, closely related to quantal response equilibrium and logit-based evolutionary or learning dynamics \citep{McKelveyPalfrey1995QRE,GoereeHoltPalfrey2005RegularQRE,Blume1993LogitDynamics}. A principled foundation arises from the Gibbs variational principle: soft responses emerge when an agent maximizes expected payoff while incurring an explicit information-processing cost, with the inverse temperature parameter governing decision precision and resource expenditure \citep{OrtegaBraun2013ThermodynamicsDecision}.
This temperature-based interpretation aligns with stochastic-choice and rational-inattention formulations \citep{Sims2003RationalInattention,MatejkaMcKay2015RI}, underpinning smooth learning dynamics based on logit-type updates \citep{HofbauerSandholm2002StochasticFP}. Similarly, other significant deviations from perfect rationality emphasize distinct mechanisms, such as limited strategic depth, systematic belief misspecification, and explicit computational constraints; these models are not always representable by temperature parameters \citep{StahlWilson1995LevelK,CamererHoChong2004CH,EysterRabin2005Cursed,Rubinstein1986FiniteAutomata,HalpernPass2015AlgorithmicRationality}.
Our two-temperature formulation fits within the temperature-based framework, allowing for typical-case equilibrium predictions in random games. A detailed relationship to existing relaxed-game frameworks is provided in Appendix \ref{app:bounded_rationality}.

\paragraph{Replica Methods and Chain-of-Replicas.}
The replica method is a central heuristic for typical-case analysis in disordered systems \citep{edwards1975theory,mezard1987spin,mezard2009information} and has increasingly influenced modern learning theory through high-dimensional asymptotics.
In statistical physics, the chain-of-replicas construction provides an equilibrium route to sequential structure; it was originally introduced to represent slow relaxational dynamics and capture long-time limits \citep{FranzParisi2013JSTAT,ZdeborovaKrzakala2010PRB,SunCrisantiKrzakalaLeuzziZdeborova2012JSTAT,KrzakalaKurchan2007PRE}.
Recently, replica-chain concepts have been adapted for multi-stage and iterative learning procedures, where each stage addresses an optimization problem based on earlier stages, as traditional dynamical formalisms are not directly applicable. Representative examples include replica-chain analyses of alternating minimization and other staged pipelines, such as distillation, transfer learning, self-training, and generative adversarial networks \citep{OkajimaTakahashi2025JSTAT,pmlr-v145-saglietti22a,OkajimaObuchi2025TMLR,Takahashi2022ArXiv,takanami2025the, IchikawaHukushima2025TMLR}.
In our min-max setting, the analysis yields a nested replica structure that can be interpreted as a short replica chain connecting two thermodynamic layers. 
We leverage this structure to characterize typical equilibrium rather than derive an explicit time-evolution description.

\section{Conclusion}
\label{sec:conclusion}
This work examines strategic-form zero-sum min-max games from a typical-case perspective under bounded rationality. We introduce thermal min-max games, a two-temperature thermodynamic relaxation where each player is assigned an inverse temperature controlling the sharpness of their response. This leads to an asymmetric yet interpretable model that smoothly interpolates between diffuse play and near-best responses, illustrating how heterogeneous computational resources shape equilibrium behavior.
Building on this formulation, we develop a two-temperature nested replica framework that renders typical-case analysis tractable in the large-strategy regime. For Gaussian payoff matrices, the resulting saddle-point characterization provides explicit predictions for the typical equilibrium value and mixed-strategy statistics, including concentration and sparsity, as functions of the two rationality parameters, the strategy-count aspect ratio, and payoff randomness. Taking the ordered zero-temperature limit recovers the classical Nash min-max value and provides a low-dimensional description of typical equilibria.
Numerical experiments demonstrate that these asymptotic predictions closely match equilibria in finite games of moderate size, highlighting the practical utility of the thermodynamic approximation. More broadly, this framework offers a principled toolkit for examining \emph{universality in games}, enabling systematic comparisons across payoff ensembles and model variants, and identifying robust equilibrium features as distributional details change.

\section*{Acknowledgements}
We thank Koji Hukushima, Kai Nakaishi, and Yasushi Nagano for their valuable comments and discussions.

\bibliographystyle{unsrtnat}
\bibliography{ref}

\appendix
\newpage

\section{Derivation}
\label{app-sec:derivation}

\subsection{Replicated partition function and Gaussian average}

Define
\begin{equation}
Z(\B{C})=\int_{\mac{X}_N}d \B{x}[Z_y(\B{x};\B{C})]^{k},
~~
Z_y(\B{x};\B{C})=\int_{\mac{Y}_M}d\B{y}\exp\ab(\beta_{\max}E(\B{x},\B{y};\B{C})).
\end{equation}
For the replica calculation, temporarily assume $k\in\mab{N}$. Replicating the $\B{x}$-integral $n$ times and the $\B{y}$-integral $k$ times within each replica gives
\begin{multline}
\mab{E}[Z^n]
=
\int\ab(\prod_{a=1}^n d\B{x}^a\1\{\B{x}^a\in\mac{X}_N\})
\ab(\prod_{a=1}^{n}\prod_{l=1}^{k} d\B{y}^{al}\1\{\B{y}^{al}\in\mac{Y}_M\})\\
\times \mab{E}\ab[\exp\ab(\beta_{\max}\sum_{a=1}^n\sum_{l=1}^k E(\B{x}^a,\B{y}^{al};\B{C}))].    
\end{multline}

For the Gaussian bilinear game,
\begin{equation}
\sum_{a=1}^n\sum_{l=1}^k E(\B{x}^a,\B{y}^{al};\B{C})
=
\sqrt{\frac{\sigma}{L}}\sum_{i=1}^N\sum_{j=1}^M C_{ij}
\ab(\sum_{a=1}^n\sum_{l=1}^k x_i^a y_j^{al}).
\end{equation}
Evaluating Gaussian integral with
$A_{ij}=\beta_{\max}\sqrt{\sigma/L}\sum_{a,l}x_i^a y_j^{al}$ yields
\begin{equation}
\mab{E}_{\B{C}}\ab[\exp\ab(\beta_{\max}\sqrt{\frac{\sigma}{L}}
\sum_{i=1}^N\sum_{j=1}^M C_{ij}\sum_{a=1}^n\sum_{l=1}^k x_i^a y_j^{al})]
=
\exp\ab(\frac{\sigma\beta_{\max}^2}{2L}\sum_{i=1}^N\sum_{j=1}^M
\ab(\sum_{a=1}^n\sum_{l=1}^k x_i^a y_j^{al})^2).
\end{equation}
Expanding the square gives
\begin{equation}
\sum_{i=1}^N\sum_{j=1}^M
\ab(\sum_{a=1}^n\sum_{l=1}^k x_i^a y_j^{al})^2
=
\sum_{a=1}^n\sum_{b=1}^n\sum_{l=1}^k\sum_{s=1}^k
\ab(\sum_{i=1}^N x_i^a x_i^b)
\ab(\sum_{j=1}^M y_j^{al}y_j^{bs}).
\end{equation}

\subsection{Overlap order parameters and site factorization}

Introduce the overlaps
\begin{equation}
Q^x_{ab}=\frac{1}{N}\sum_{i=1}^N x_i^a x_i^b,
~~
Q^y_{abls}=\frac{1}{M}\sum_{j=1}^M y_j^{al}y_j^{bs}.
\end{equation}
Then
\begin{equation}
\sum_{i=1}^N\sum_{j=1}^M
\ab(\sum_{a=1}^n\sum_{l=1}^k x_i^a y_j^{al})^2
=
NM\sum_{a=1}^n\sum_{b=1}^n\sum_{l=1}^k\sum_{s=1}^k
Q^x_{ab}Q^y_{abls}.
\end{equation}

We enforce the overlap definitions using Fourier representations of delta constraints (up to subexponential normalization factors):
\begin{equation}
1\asymp
\int d Q^xd\tilde Q^x
\exp\ab(\frac{N}{2}\sum_{a=1}^n\sum_{b=1}^n \tilde Q^x_{ab}
\ab(Q^x_{ab}-\frac{1}{N}\sum_{i=1}^N x_i^a x_i^b)),
\end{equation}
\begin{equation}
1\asymp
\int d Q^yd\tilde Q^y
\exp\ab(\frac{M}{2}\sum_{a=1}^n\sum_{b=1}^n\sum_{l=1}^k\sum_{s=1}^k \tilde Q^y_{abls}
\ab(Q^y_{abls}-\frac{1}{M}\sum_{j=1}^M y_j^{al}y_j^{bs})).
\end{equation}

For the simplex constraints we use Lemma~\ref{lem:simplex-delta-trunc}. For each outer replica $a$,
\begin{equation}
\1\{\B{x}^a\in\mac{X}_N\}
=
\ab(\prod_{i=1}^N\1\{0\le x_i^a\le N\})\delta\ab(N-\sum_{i=1}^N x_i^a),
\end{equation}
and for each $(a,l)$,
\begin{equation}
\1\{\B{y}^{al}\in\mac{Y}_M\}
=
\ab(\prod_{j=1}^M\1\{0\le y_j^{al}\le M\})\delta\ab(M-\sum_{j=1}^M y_j^{al}).
\end{equation}
We represent each delta by Lemma~\ref{lem:delta-fourier}:
\begin{equation}
\delta\ab(N-\sum_{i=1}^N x_i^a)
=
\int_{i\mab{R}}\frac{d m_a}{2\pi i}
\exp\ab(-m_a\ab(N-\sum_{i=1}^N x_i^a)),
\end{equation}
\begin{equation}
\delta\ab(M-\sum_{j=1}^M y_j^{al})
=
\int_{i\mab{R}}\frac{d \mu_{al}}{2\pi i}
\exp\ab(-\mu_{al}\ab(M-\sum_{j=1}^M y_j^{al})).
\end{equation}

Collecting terms, the replicated integral factorizes into one-site contributions:
\begin{multline}
    \mab{E}[Z^n]\asymp
\int d \Theta 
\ab(\prod_{a=1}^n\int_{i\mab{R}}d m_a)
\ab(\prod_{a=1}^n\prod_{l=1}^k\int_{i\mab{R}}d \mu_{al}) 
\exp\ab(\frac{\sigma\beta_{\max}^2 NM}{2L}\sum_{a,b,l,s}Q^x_{ab}Q^y_{abls}) \\
\exp\ab(\frac{N}{2}\sum_{a,b}\tilde Q^x_{ab}Q^x_{ab}+\frac{M}{2}\sum_{a,b,l,s}\tilde Q^y_{abls}Q^y_{abls}-N\sum_a m_a-M\sum_{a,l}\mu_{al})
[I_x]^N[I_y]^M,
\end{multline}
where $d\Theta \coloneqq d Q^xd\tilde Q^xd Q^yd\tilde Q^y$
\begin{equation}
I_x=\int_{[0,N]^n}\ab(\prod_{a=1}^n d x^a)
\exp\ab(-\frac{1}{2}\sum_{a=1}^n\sum_{b=1}^n \tilde Q^x_{ab}x^a x^b+\sum_{a=1}^n m_a x^a),
\end{equation}
\begin{equation}
I_y=\int_{[0,M]^{nk}}\ab(\prod_{a=1}^n\prod_{l=1}^k d y^{al})
\exp\ab(-\frac{1}{2}\sum_{a=1}^n\sum_{b=1}^n\sum_{l=1}^k\sum_{s=1}^k \tilde Q^y_{abls}y^{al}y^{bs}+\sum_{a=1}^n\sum_{l=1}^k \mu_{al}y^{al}).
\end{equation}

\subsection{RS/1RSB ansatz and Hubbard--Stratonovich decoupling}
\label{app-subsec:rs-rsb}
To compute the typical free energy using the replica method, we require a tractable parametrization of the overlap matrices that respects the nested replica structure induced by the two-temperature construction. We therefore impose replica symmetry (RS) across the \emph{outer} replicas associated with the minimizer variables $x$. The maximizer variables $y$ carry a nested index $(a,l)$, where $a$ denotes the outer replica and $l$ the inner one. Under RS at the outer level, the overlaps among $y$-replicas naturally organize into a \emph{one-step replica-symmetry-breaking} (1RSB) structure within each fixed $a$: replicas sharing the same outer index form a cluster, while the inner replicas fluctuate within that cluster. Figure~\ref{fig_rs_1rsb_ansatz} illustrates this RS/1RSB hierarchy.

\begin{assumption}
\label{assump:rs-1rsb}
We assume replica symmetry (RS) for the $x$-overlaps across outer replicas $a,b\in\{1,\dots,n\}$:
\begin{equation}
Q^x_{ab}= \begin{cases}
Q_x & (a=b),\\
q_x & (a\neq b),
\end{cases}
~~
\tilde{Q}^x_{ab}= \begin{cases}
\hat{Q}_x-\hat{\chi}_x & (a=b),\\
-\hat{\chi}_x & (a\neq b),
\end{cases}
~~
m_a=m_x.
\end{equation}
For the $y$-overlaps, indexed by $(a,l)$ with $l,s\in\{1,\dots,k\}$, we assume the corresponding induced one-step replica-symmetry-breaking (1RSB) form:
\begin{equation}
Q^y_{abls}=
\begin{cases}
Q_y & (a=b,\ l=s),\\
q_1 & (a=b,\ l\neq s),\\
q_0 & (a\neq b),
\end{cases}
~~
\tilde{Q}^y_{abls}= \begin{cases}
\hat{Q}_y-\hat{\chi}_1-\hat{\chi}_0 & (a=b,\ l=s),\\
-\hat{\chi}_1-\hat{\chi}_0 & (a=b,\ l\neq s),\\
-\hat{\chi}_0 & (a\neq b),
\end{cases}
~~
\mu_{al}=m_y.
\end{equation}
\end{assumption}

\begin{figure}[tb] 
    \centering \includegraphics[width=\linewidth]{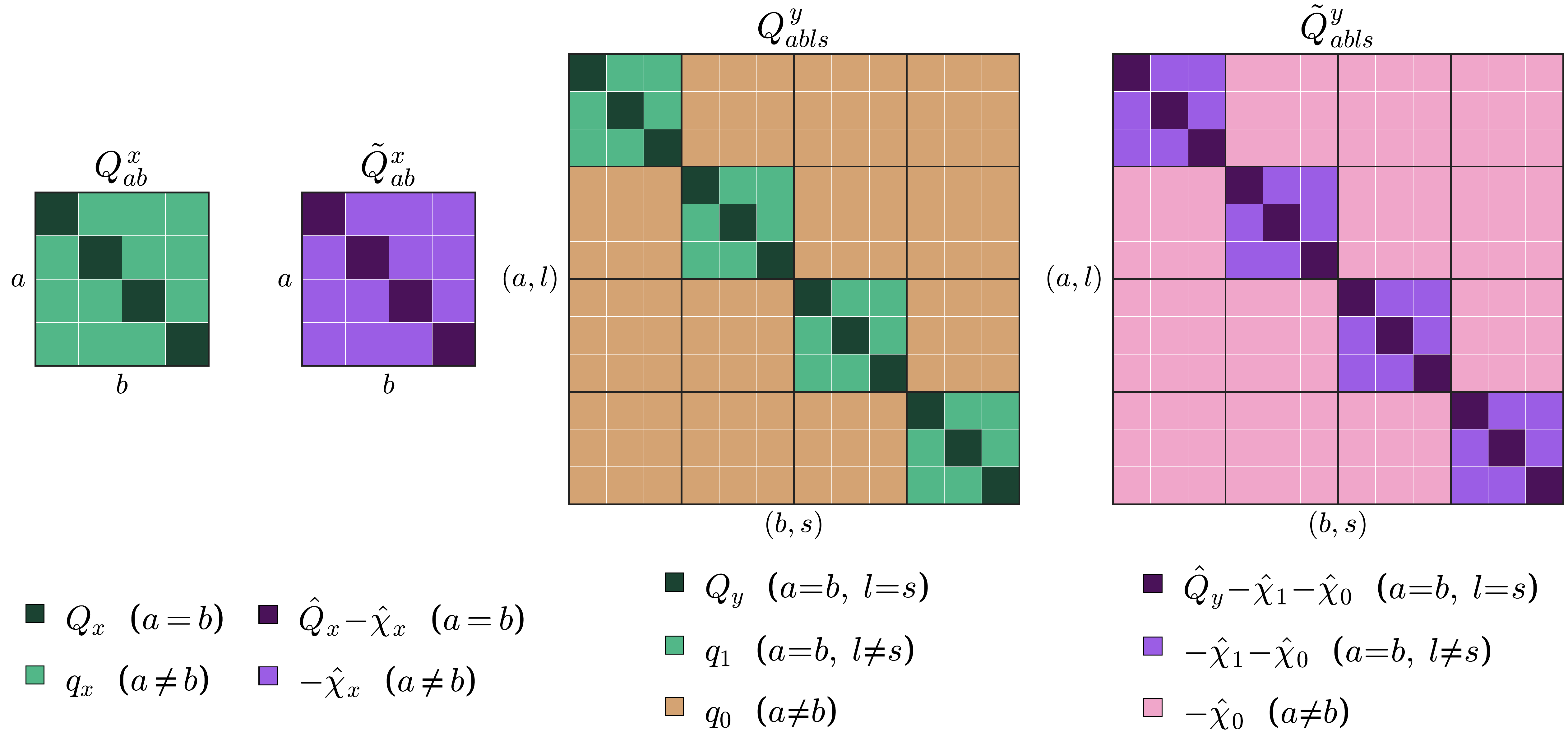} \caption{Schematic of the RS/1RSB hierarchy induced by the nested replica indices: RS across outer replicas $a$ for $x$, and a 1RSB block structure in the inner replicas $l$ for each fixed $a$ for $y$.} \label{fig_rs_1rsb_ansatz} \end{figure}

With this ansatz, the one-site quadratic forms become rank-one perturbations and can be decoupled by Lemma~\ref{lem:hs}.

\paragraph{$x$-site.}
Under RS, $\tilde Q^x=\hat{Q}_x I-\hat\chi_x \1\1^\top$, so
\begin{equation}
-\frac{1}{2}\sum_{a=1}^n\sum_{b=1}^n \tilde Q^x_{ab}x_a x_b + m_x\sum_{a=1}^n x_a
=
-\frac{\hat{Q}_x}{2}\sum_{a=1}^n x_a^2
+\frac{\hat\chi_x}{2}\ab(\sum_{a=1}^n x_a)^2
+m_x\sum_{a=1}^n x_a.
\end{equation}
By Lemma~\ref{lem:hs},
\begin{equation}
\exp\ab(\frac{\hat\chi_x}{2}\ab(\sum_{a=1}^n x_a)^2)
=
\int Dz\exp\ab(\sqrt{\hat\chi_x}z\sum_{a=1}^n x_a),
\end{equation}
hence the $x$-site integral factorizes:
\begin{equation}
\int_{[0,N]^n}\ab(\prod_{a=1}^n d x_a)
\exp\ab(
-\frac{\hat{Q}_x}{2}\sum_{a=1}^n x_a^2+\ab(m_x+\sqrt{\hat\chi_x}z)\sum_{a=1}^n x_a
)
=
[Z_x(z)]^n,
\end{equation}
where
\begin{equation}
Z_x(z)=\int_{0}^{N}d x
\exp\ab(-\frac{\hat{Q}_x}{2}x^2+\ab(m_x+\sqrt{\hat\chi_x}z)x).
\end{equation}

\paragraph{$y$-site.}
Under 1RSB, $\tilde Q^y_{abls}=\hat{Q}_y\delta_{ab}\delta_{ls}-\hat\chi_1\delta_{ab}-\hat\chi_0$, and
\begin{equation}
-\frac{1}{2}\sum_{a,b,l,s}\tilde Q^y_{abls}y_{al}y_{bs}
=
-\frac{\hat{Q}_y}{2}\sum_{a=1}^n\sum_{l=1}^k y_{al}^2
+\frac{\hat\chi_1}{2}\sum_{a=1}^n\ab(\sum_{l=1}^k y_{al})^2
+\frac{\hat\chi_0}{2}\ab(\sum_{a=1}^n\sum_{l=1}^k y_{al})^2.
\end{equation}
We decouple the $\hat\chi_0$-term by Lemma~\ref{lem:hs} with a single $z\sim Dz$:
\begin{equation}
\exp\ab(\frac{\hat\chi_0}{2}\ab(\sum_{a,l}y_{al})^2)
=
\int Dz\exp\ab(\sqrt{\hat\chi_0}z\sum_{a,l}y_{al}).
\end{equation}
Importantly, the $\hat\chi_1$-term must be decoupled separately for each outer replica $a$:
\begin{equation}
\exp\ab(\frac{\hat\chi_1}{2}\ab(\sum_{l=1}^k y_{al})^2)
=
\int D\eta\exp\ab(\sqrt{\hat\chi_1}\eta\sum_{l=1}^k y_{al}).
\end{equation}
With $\mu_{al}=m_y$, the $y$-site integral factorizes over $l$ inside each $a$:
\begin{equation}
\int\ab(\prod_{a=1}^n\prod_{l=1}^k d y_{al})
\exp\ab(
-\frac{\hat{Q}_y}{2}\sum_{a,l} y_{al}^2
+\sum_{a,l}\ab(m_y+\sqrt{\hat\chi_0}z+\sqrt{\hat\chi_1}\eta_a)y_{al}
)
=
\ab[\int D\eta Z_y(z,\eta)^k]^n,
\end{equation}
where
\begin{equation}
Z_y(z,\eta)=\int_{0}^{M}d y
\exp\ab(-\frac{\hat{Q}_y}{2}y^2+\ab(m_y+\sqrt{\hat\chi_0}z+\sqrt{\hat\chi_1}\eta)y).
\end{equation}

\subsection{The $n\to 0$ limit}

Applying Lemma~\ref{lem:replica-log} with $A(z)=Z_x(z)$ and $A(z)=\int D\eta Z_y(z,\eta)^k$ gives
\begin{equation}
\lim_{n\to 0}\frac{1}{n}\log\int Dz[Z_x(z)]^n=\int Dz\log Z_x(z),
\end{equation}
\begin{equation}
\lim_{n\to 0}\frac{1}{n}\log\int Dz\ab[\int D\eta Z_y(z,\eta)^k]^n
=
\int Dz\log\int D\eta Z_y(z,\eta)^k.
\end{equation}

\subsection{Moment Equations}
\label{app:moment-deriv-subsec}

We derive the conjugate relations and moment equations used in Claim~\ref{thm:finiteT}.
The key technical points are the chain rule through $h=m+\sqrt{\hat\chi}(\cdot)$ and Gaussian integration by parts (Lemma~\ref{lem:gauss-ibp}).

\paragraph{One-site Measures.}
Define
\begin{equation}
h_x(z)=m_x+\sqrt{\hat\chi_x}z,
~~
h_y(z,\eta)=m_y+\sqrt{\hat\chi_0}z+\sqrt{\hat\chi_1}\eta.
\end{equation}
Introduce Gibbs averages on the bounded domains:
\begin{equation}
\langle A\rangle_z=\frac{1}{Z_x(z)}\int_{0}^{N}d xA(x)\exp\ab(-\frac{\hat{Q}_x}{2}x^2+h_x(z)x),
\end{equation}
\begin{equation}
\langle A\rangle_{z,\eta}=\frac{1}{Z_y(z,\eta)}\int_{0}^{M}d yA(y)\exp\ab(-\frac{\hat{Q}_y}{2}y^2+h_y(z,\eta)y).
\end{equation}
For the 1RSB reweighting, set
\begin{equation}
\Psi(z)=\int D\eta Z_y(z,\eta)^k,
~~
\langle A\rangle_{\eta|z}=\frac{1}{\Psi(z)}\int D\eta Z_y(z,\eta)^kA(z,\eta).
\end{equation}

\paragraph{Conjugate Relations from Polynomial Terms.}
Differentiating the RS/1RSB saddle functional $g$ with respect to $q_x,q_0,q_1,Q_y$ yields
\begin{equation}
\hat\chi_x=\frac{\sigma\beta_{\max}^2}{\gamma^{\nicefrac{1}{2}}}k^2 q_0,
~~
\hat\chi_0=\gamma^{\nicefrac{1}{2}}\sigma\beta_{\max}^2 q_x,
~~
\hat\chi_1=\gamma^{\nicefrac{1}{2}}\sigma\beta_{\max}^2(Q_x-q_x),
~~
\hat{Q}_y=0.
\end{equation}

\paragraph{Constraints from $m_x$ and $m_y$.}
Using $\partial_{m_x}\log Z_x(z)=\langle x\rangle_z$ and $\partial_{m_y}\log Z_y(z,\eta)=\langle y\rangle_{z,\eta}$,
stationarity gives the simplex constraints in the one-site representation:
\begin{equation}
\int Dz\langle x\rangle_z=1,
~~
\int Dz\langle\langle y\rangle_{z,\eta}\rangle_{\eta|z}=1.
\end{equation}

\paragraph{Second Moments from $\hat{Q}_x$ and $\hat{Q}_y$.}
Since $\partial_{\hat{Q}_x}\log Z_x(z)=-\frac{1}{2}\langle x^2\rangle_z$,
\begin{equation}
Q_x=\int Dz\langle x^2\rangle_z.
\end{equation}
Similarly, $\partial_{\hat{Q}_y}\log Z_y(z,\eta)=-\frac{1}{2}\langle y^2\rangle_{z,\eta}$ implies
\begin{equation}
\frac{\partial}{\partial \hat{Q}_y}\log\Psi(z)=-\frac{k}{2}\langle\langle y^2\rangle_{z,\eta}\rangle_{\eta|z},
\end{equation}
Stationarity yields
\begin{equation}
Q_y=\int Dz\langle\langle y^2\rangle_{z,\eta}\rangle_{\eta|z}.
\end{equation}

\paragraph{Moment equation for $q_x$.}
By the chain rule,
\begin{equation}
\frac{\partial}{\partial \hat\chi_x}\log Z_x(z)=\frac{z}{2\sqrt{\hat\chi_x}}\langle x\rangle_z.
\end{equation}
Stationarity in $\hat\chi_x$ gives
\begin{equation}
Q_x-q_x=\frac{1}{\sqrt{\hat\chi_x}}\int Dzz\langle x\rangle_z.
\end{equation}
Applying Lemma~\ref{lem:gauss-ibp} to $f(z)=\langle x\rangle_z$ and using
$\frac{d}{d z}\langle x\rangle_z=\sqrt{\hat\chi_x}(\langle x^2\rangle_z-\langle x\rangle_z^2)$,
we obtain
\begin{equation}
q_x=\int Dz\langle x\rangle_z^2.
\end{equation}

\paragraph{Moment equation for $q_1$.}
Similarly,
\begin{equation}
\frac{\partial}{\partial \hat\chi_1}\log Z_y(z,\eta)=\frac{\eta}{2\sqrt{\hat\chi_1}}\langle y\rangle_{z,\eta},
~~
\frac{\partial}{\partial \hat\chi_1}\log\Psi(z)=\frac{k}{2\sqrt{\hat\chi_1}}\langle \eta\langle y\rangle_{z,\eta}\rangle_{\eta|z}.
\end{equation}
The stationarity equation in $\hat\chi_1$ can be written as
\begin{equation}
Q_y+(k-1)q_1=\frac{1}{\sqrt{\hat\chi_1}}\int Dz\langle \eta\langle y\rangle_{z,\eta}\rangle_{\eta|z}.
\end{equation}
Applying Lemma~\ref{lem:gauss-ibp} in the $\eta$-variable under $D\eta$ and using
$\partial_\eta \log Z_y(z,\eta)=\sqrt{\hat\chi_1}\langle y\rangle_{z,\eta}$,
one obtains
\begin{equation}
\frac{1}{\sqrt{\hat\chi_1}}\langle \eta\langle y\rangle_{z,\eta}\rangle_{\eta|z}
=
\langle\langle y^2\rangle_{z,\eta}\rangle_{\eta|z}
+(k-1)\langle\langle y\rangle_{z,\eta}^2\rangle_{\eta|z}.
\end{equation}
Integrating over $z$ and using the definition of $Q_y$ gives
\begin{equation}
q_1=\int Dz\langle\langle y\rangle_{z,\eta}^2\rangle_{\eta|z}.
\end{equation}

\paragraph{Moment equation for $q_0$.}
For $\hat\chi_0$,
\begin{equation}
\frac{\partial}{\partial \hat\chi_0}\log Z_y(z,\eta)=\frac{z}{2\sqrt{\hat\chi_0}}\langle y\rangle_{z,\eta},
~~
\frac{\partial}{\partial \hat\chi_0}\log\Psi(z)=\frac{kz}{2\sqrt{\hat\chi_0}}\langle\langle y\rangle_{z,\eta}\rangle_{\eta|z}.
\end{equation}
The stationarity equation in $\hat\chi_0$ can be written as
\begin{equation}
Q_y+(k-1)q_1-kq_0=\frac{1}{\sqrt{\hat\chi_0}}\int Dzz\bar y(z),
~~
\bar y(z)=\langle\langle y\rangle_{z,\eta}\rangle_{\eta|z}.
\end{equation}
Applying Lemma~\ref{lem:gauss-ibp} to $f(z)=\bar y(z)$ and differentiating $\bar y(z)$ under the integral sign yields
\begin{equation}
\frac{1}{\sqrt{\hat\chi_0}}\frac{d}{d z}\bar y(z)
=
\langle\langle y^2\rangle_{z,\eta}\rangle_{\eta|z}
+(k-1)\langle\langle y\rangle_{z,\eta}^2\rangle_{\eta|z}
-k\bar y(z)^2.
\end{equation}
Integrating over $z$ and using $Q_y=\int Dz\langle\langle y^2\rangle_{z,\eta}\rangle_{\eta|z}$ and the definition of $q_1$ gives
\begin{equation}
\frac{1}{\sqrt{\hat\chi_0}}\int Dzz\bar y(z)=Q_y+(k-1)q_1-k\int Dz\bar y(z)^2.
\end{equation}
Comparing with the stationarity equation yields
\begin{equation}
q_0=\int Dz\bar y(z)^2
=
\int Dz\ab(\langle\langle y\rangle_{z,\eta}\rangle_{\eta|z})^2.
\end{equation}

These identities close the moment system for the $x$- and $y$-marginals in Claim~\ref{thm:finiteT}.

\subsection{Typical Payoff Density}
\label{app:payoff-from-op}

This subsection shows how the disorder-averaged payoff density can be expressed in closed form using the RS/1RSB saddle-point order parameters.
Throughout we use the Gaussian scaling
$E(\B{x},\B{y};\B{C})=\sqrt{\sigma/L}\sum_{i,j}C_{ij}x_i y_j$ and the ratio $k=-\nicefrac{\beta_{\min}}{\beta_{\max}}$.
Define the (typical) payoff density under the two-temperature Boltzmann strategies (TTBS) by
\begin{equation}
e(\beta_{\max},\beta_{\min})
:=
\lim_{L\to\infty}\frac{1}{L}
\mab E_{\B{C}}\ab[\Big\langle E(\B{X},\B{Y};\B{C})\Big\rangle_{\beta_{\max},\beta_{\min}}],
\label{eq:payoff-density-def}
\end{equation}
where $\langle\cdot\rangle_{\beta_{\max},\beta_{\min};\B{C}}$ denotes the expectation with respect to the TTBS sampling.

For a fixed instance $\B{C}$, differentiating the partition function
$Z(\B{C})=\int_{\mac{X}_N}d\B{X}[Z_y(\B{X};\B{C})]^{k}$
(with $Z_y(\B{X};\B{C})=\int_{\mac{Y}_M}d\B{Y}\exp(\beta_{\max}E(\B{X},\B{Y};\B{C}))$)
and using $\partial_\sigma F=E/(2\sigma)$ yields the exact identity
\begin{equation}
\frac{\partial}{\partial\sigma}\log Z(\B{C})
=
\frac{k\beta_{\max}}{2\sigma}
\Big\langle E(\B{X},\B{Y};\B{C})\Big\rangle_{\beta_{\max},\beta_{\min};\B{C}}.
\label{eq:dlogZ-dsigma-payoff}
\end{equation}
Recalling $\Phi(\beta_{\max},\beta_{\min};\B{C})=-(1/\beta_{\min})\log Z(\B{C})$ and
$\nu(\beta_{\max},\beta_{\min})=\lim_{L\to\infty}\frac{1}{L}\mab E_{\B{C}}[\Phi(\beta_{\max},\beta_{\min};\B{C})]$,
and substituting $k=-\nicefrac{\beta_{\min}}{\beta_{\max}}$, we derive the compact relation
\begin{equation}
e(\beta_{\max},\beta_{\min})
=
2\sigma\frac{\partial}{\partial\sigma}\nu(\beta_{\max},\beta_{\min}).
\label{eq:payoff-from-v}
\end{equation}

Finally, at finite temperature, Claim~\ref{thm:finiteT} gives
$\nu(\beta_{\max},\beta_{\min})=-(1/\beta_{\min})\mathrm{extr}_{\Theta}g(\Theta;\sigma)$.
Since $\sigma$ enters $g$ only through the explicit overlap-polynomial term, the envelope rule at the stationary point implies that
$\partial_\sigma v$ is obtained by differentiating only this explicit term and evaluating it at the RS/1RSB saddle.
This yields the closed-form payoff density
\begin{equation}
e(\beta_{\max},\beta_{\min})
=
\sigma\beta_{\max}\Big(
Q_xQ_y+(k-1)Q_xq_1-kq_xq_0
\Big),
\label{eq:payoff-from-ops}
\end{equation}
where $(Q_x,q_x)$ are the RS overlaps for the $\B{X}$ replicas and $(Q_y,q_1,q_0)$ are the induced 1RSB overlaps for the $\B{Y}$ replicas, all evaluated at the saddle point.

\section{Ordered Zero-Temperature Limit}
\label{app:oztl-detailed}

We derive the ordered zero-temperature limit of the saddle-point equations from the finite-temperature free energy in Claim~\ref{thm:finiteT}.
Recall $k=-\nicefrac{\beta_{\min}}{\beta_{\max}}<0$ and consider the ordered limit, $\beta_{\max}\to\infty$ first, then $\beta_{\min}\to\infty$, $k\to 0^-$, Although $k\to 0^-$, the $y$-site contribution has $\log Z_y(z,\eta)=\Theta(\beta_{\max})$, so that $k\log Z_y=\Theta(\beta_{\min})$ survives and induces a second Laplace principle at scale $\beta_{\min}$.

\subsection{Ordered Scaling}
\label{app:oztl-scaling}
To take the ordered zero temperature limit, we scale the order parameters as follows:
\begin{equation}
    Q_x-q_x=\frac{\chi_x}{\beta_{\min}}+o(\beta_{\min}^{-1}),
    ~~
    q_1-q_0=\frac{\Delta}{\beta_{\min}}+o(\beta_{\min}^{-1}),
    ~~
    Q_y-q_1=\frac{\chi}{\beta_{\max}}+o(\beta_{\max}^{-1}).
    \label{eq:collision-ansatz}
\end{equation}
Equivalently,
\begin{equation}
Q_x=q_x+\frac{\chi_x}{\beta_{\min}},
~~
q_0=q,
~~
q_1=q+\frac{\Delta}{\beta_{\min}},
~~
Q_y=q+\frac{\Delta}{\beta_{\min}}+\frac{\chi}{\beta_{\max}}.
\label{eq:collision-expanded}
\end{equation}
We scale conjugates so that the $x$-site exponent is $\Theta(\beta_{\min})$ while the inner $y$-site exponent is $\Theta(\beta_{\max})$.
With a slight abuse of notation, we reuse the same symbols for the $\mac{O}(1)$ rescaled quantities:
\begin{align}
    &\hat{Q}_x=\beta_{\min}\hat{Q}_x,
    ~~
    \hat{\chi}_x=\beta_{\min}^2\hat{\chi}_x^{(0)},
    ~~
    m_x=\beta_{\min}\hat{m}_x,
    \label{eq:scaling-x-conj} \\
    &\hat{Q}_y=\beta_{\max}\hat{Q}_y,
    ~~
    \hat{\chi}_0=\beta_{\max}^2\hat{\chi},
    ~~
    \hat{\chi}_1=\frac{\beta_{\max}^2}{\beta_{\min}}\hat{\Delta},
    ~~
    m_y=\beta_{\max}\hat{m}_y.
\label{eq:scaling-y-conj}
\end{align}
We assume a priori that $\hat{Q}_y=\Theta(\beta_{\max})$, since the finite-$k$ identity $\hat{Q}_y=0$ is obtained by dividing by $k$, which is not valid as $k\to 0^-$. As shown in Appendix~\ref{app:oztl-conjugate}, the ordered saddle forces $\hat{Q}_y=0$.

\subsection{$\beta_{\min}$-Laplace Principle}
\label{app:oztl-xsite}

Under Eq.~\eqref{eq:scaling-x-conj}, the one-site partition function is
\begin{equation}
    Z_x(z)=\int_{0}^{N}dx
    \exp\ab[
    -\beta_{\min}\ab(
    \frac{\hat{Q}_x}{2}x^2-\ab(\hat{m}_x+\sqrt{\hat{\chi}_x^{(0)}}z)x)].
    \label{eq:Zx-ordered}
\end{equation}
Since the domain $[0,N]$ is compact, Lemma~\ref{lem:laplace-compact} yields
\begin{equation}
    \frac{1}{\beta_{\min}}\log Z_x(z)\to
    \sup_{x\in[0,N]}\ab(h_x(z)x-\frac{\hat{Q}_x}{2}x^2),
    ~~~h_x(z)=\hat{m}_x+\sqrt{\hat{\chi}_x^{(0)}}z.
\label{eq:laplace-x-clean}
\end{equation}
At the ordered saddle, the maximizer is $\mac{O}(1)$; hence, the boundary $x=N$ is asymptotically inactive, and the constrained maximizer is $x^\star(z)=\pos{\nicefrac{h_x(z)}{\hat{Q}_x}}$.
Define
\begin{equation}
    \alpha_x=\hat{m}_x(\hat{\chi}_x^{(0)})^{-\nicefrac{1}{2}},
    ~~
    A(\alpha)=\int Dz \pos{z+\alpha},
    ~~
    B(\alpha)=\int Dz \pos{z+\alpha}^2,
    ~~
    q(\alpha)=\frac{B(\alpha)}{A(\alpha)^2}.
    \label{eq:ABq-def-clean}
\end{equation}
The simplex constraint $\int Dz x^\star(z)=1$ gives $\hat{Q}_x=(\hat{\chi}_x^{(0)})^{\nicefrac{1}{2}}A(\alpha_x)$,
and therefore the canonical form
\begin{equation}
x^\star(z)=\frac{\pos{z+\alpha_x}}{A(\alpha_x)}.
\label{eq:x-canonical-clean}
\end{equation}
Consequently,
\begin{equation}
q_x=\int Dz(x^\star(z))^2=q(\alpha_x),
~~
\rho_x=\Pr(x^\star(Z)>0)=\Phi(\alpha_x).
\label{eq:x-moments-support-clean}
\end{equation}

\subsection{$\beta_{\max}$-Laplace Principle}
\label{app:oztl-ysite}

The inner one-site partition function is
\begin{equation}
    Z_y(z,\eta)=\int_{0}^{M}dy 
    \exp\ab(
    -\frac{\hat{Q}_y}{2}y^2+\ab(m_y+\sqrt{\hat{\chi}_0}z+\sqrt{\hat{\chi}_1}\eta)y),
    \label{eq:Zy-start-clean}
\end{equation}
and it enters only through
\begin{equation}
    \Psi(z)=\int D\eta[Z_y(z,\eta)]^k,
    ~~
    k=-\frac{\beta_{\min}}{\beta_{\max}},~~ h(z,\eta)=\hat{m}_y+\sqrt{\hat{\chi}}z+\sqrt{\frac{\hat{\Delta}}{\beta_{\min}}}\eta.
    \label{eq:Psi-def-clean}
\end{equation}
Applying Lemma~\ref{lem:laplace-compact} to the compact interval $[0,M]$ and noting that the maximizer is $\mac{O}(1)$ at the ordered saddle,
\begin{equation}
    \log Z_y(z,\eta)=\beta_{\max}\sup_{y\ge 0}\ab(h(z,\eta)y-\frac{\hat{Q}_y}{2}y^2)+o(\beta_{\max})
    =\frac{\beta_{\max}}{2\hat{Q}_y}\pos{h(z,\eta)}^2+o(\beta_{\max}).
    \label{eq:logZy-laplace-clean}
\end{equation}
Using $k\beta_{\max}=-\beta_{\min}$,
\begin{equation}
    [Z_y(z,\eta)]^k
    =
    \exp\ab(
    -\frac{\beta_{\min}}{2\hat{Q}_y}\pos{h(z,\eta)}^2+o(\beta_{\min})).
    \label{eq:Zy-power-k-clean}
\end{equation}
Set $\eta=\sqrt{\beta_{\min}}\xi$. Then
\begin{equation}
    \Psi(z)
    =
    \sqrt{\beta_{\min}}
    \int_{\mab{R}}\frac{d\xi}{\sqrt{2\pi}}
    \exp\ab\{
    -\beta_{\min}\ab[
    \frac{\xi^2}{2}
    +
    \frac{1}{2\hat{Q}_y}\pos{a(z)+\sqrt{\hat{\Delta}}\xi}^2]
    +o(\beta_{\min})\},
    \label{eq:Psi-xi-clean}
\end{equation}
where $a(z)=\hat{m}_y+\sqrt{\hat{\chi}}z$. The prefactor $\sqrt{\beta_{\min}}$ is negligible at scale $\beta_{\min}$, and Lemma~\ref{lem:laplace-compact} applies.
\begin{equation}
    \frac{1}{\beta_{\min}}\log\Psi(z)\to -\inf_{\xi\in\mab{R}}
    \ab[
    \frac{\xi^2}{2}
    +
    \frac{1}{2\hat{Q}_y}\pos{a(z)+\sqrt{\hat{\Delta}}\xi}^2].
    \label{eq:laplace-Psi-clean}
\end{equation}
The infimum represents a one-dimensional convex problem, and the minimizer resides in the active region when $a(z)>0$, providing
\begin{equation}
    \inf_{\xi\in\mab{R}}
    \ab[
    \frac{\xi^2}{2}
    +
    \frac{1}{2\hat{Q}_y}\pos{a(z)+\sqrt{\hat{\Delta}}\xi}^2
    ]
    =
    \frac{1}{2(\hat{Q}_y+\hat{\Delta})}\pos{a(z)}^2.
    \label{eq:inf-Fz-final-clean}
\end{equation}
Therefore,
\begin{equation}
    \frac{1}{\beta_{\min}}\log\Psi(z)
    \to
    -\frac{1}{2(\hat{Q}_y+\hat{\Delta})}\pos{\hat{m}_y+\sqrt{\hat{\chi}}z}^2.
    \label{eq:logPsi-limit-clean}
\end{equation}
The right-hand side of Eq.~\eqref{eq:logPsi-limit-clean} equals $-\sup_{y\ge 0}(a(z)y-(\hat{Q}_y+\hat{\Delta})y^2/2)$; therefore,
\begin{equation}
    y^\star(z)=\frac{\pos{\hat{m}_y+\sqrt{\hat{\chi}}z}}{\hat{Q}_y+\hat{\Delta}}.
    \label{eq:ystar-clean}
\end{equation}
Define $\alpha_y=\nicefrac{\hat{m}_y}{\sqrt{\hat{\chi}}}$. The simplex constraint $\int Dz y^\star(z)=1$ yields $\hat{Q}_y+\hat{\Delta}=\sqrt{\hat{\chi}}A(\alpha_y)$ and, therefore,
\begin{equation}
    y^\star(z)=\frac{\pos{z+\alpha_y}}{A(\alpha_y)}.
    \label{eq:y-canonical-clean}
\end{equation}
In particular,
\begin{equation}
    q_y=\int Dz(y^\star(z))^2=q(\alpha_y),
    ~~
    \rho_y=\Pr(y^\star(Z)>0)=\Phi(\alpha_y).
    \label{eq:y-moments-support-clean}
\end{equation}

\subsection{Quadratic Fluctuations and Susceptibilities}
\label{app:oztl-suscept}

On the active set, both $x$ and $y$ are interior maximizers with curvatures $\hat{Q}_x$ and $\hat{Q}_y+\hat{\Delta}$, respectively.
A second-order Laplace expansion therefore gives
\begin{equation}
    \chi_x=\beta_{\min}(Q_x-q_x)=\frac{\rho_x}{\hat{Q}_x},
\label{eq:chi-x-clean}
\end{equation}
and, for the outer-scale combination surviving the ordered limit,
\begin{equation}
    \chi-\Delta=\frac{\rho_y}{\hat{Q}_y+\hat{\Delta}}.
    \label{eq:chi-minus-Delta-clean}
\end{equation}

\subsection{Ordered-limit Conjugate Relations}
\label{app:oztl-conjugate}

We take the ordered limit of the finite-temperature stationarity conditions of $g$ from Claim~\ref{thm:finiteT}.

At finite temperature,
\begin{equation}
    \hat{\chi}_x= \gamma^{-\nicefrac{1}{2}}\sigma\beta_{\max}^2 k^2 q_0,
    ~~
    \hat{\chi}_0=\gamma^{\nicefrac{1}{2}}\sigma\beta_{\max}^2q_x,
    ~~\hat{\chi}_1=\gamma^{\nicefrac{1}{2}}\sigma\beta_{\max}^2(Q_x-q_x).
\end{equation}
Using $k^2=\nicefrac{\beta_{\min}^2}{\beta_{\max}^2}$, $q_0=q$, and Eqs.~\eqref{eq:scaling-x-conj}--\eqref{eq:scaling-y-conj},
\begin{equation}
    \hat{\chi}_x^{(0)}= \gamma^{-\nicefrac{1}{2}} \sigma q,
    ~~
    \hat{\chi}=\gamma^{\nicefrac{1}{2}} \sigma q_x,
    ~~
    \hat{\Delta}=\gamma^{\nicefrac{1}{2}} \sigma \chi_x.
    \label{eq:var-rel-ordered}
\end{equation}
At the ordered zero temperature limit, $q=q_y$ since $Q_y \to q$. Expanding $\nicefrac{\partial g}{\partial Q_x}=0$ and $\nicefrac{\partial g}{\partial Q_y}=0$ under the collision scaling Eq.~\eqref{eq:collision-expanded}
and matching orders in $\beta_{\min}$ and $\beta_{\max}$ gives
\begin{equation}
    \hat{Q}_x= \gamma^{-\nicefrac{1}{2}} \sigma(\chi-\Delta),~~\hat{Q}_y=0,
    ~~\hat{Q}_y+\hat{\Delta}=\sigma \gamma^{\nicefrac{1}{2}} \chi_x.
    \label{eq:hatqx-clean}
\end{equation}

\subsection{Support-Size Matching Law}
\label{app:oztl-support-law}

Eliminating curvatures using Eqs.~\eqref{eq:chi-x-clean}--\eqref{eq:chi-minus-Delta-clean} and Eq.~\eqref{eq:hatqx-clean} yields
\begin{equation}
    \rho_y=\gamma\rho_x.
    \label{eq:support-law-clean}
\end{equation}
Using Eq.~\eqref{eq:x-moments-support-clean} and Eq.~\eqref{eq:y-moments-support-clean}, this becomes
\begin{equation}
\Phi(\alpha_y)=\gamma\Phi(\alpha_x).
\end{equation}
All overlaps and variances are expressed in terms of $(\alpha_x,\alpha_y)$, with one scalar condition fixing the relative shift of the simplex multipliers:
\begin{equation}
    \gamma\hat{m}_x+\hat{m}_y=0.
    \label{eq:chem-balance-clean}
\end{equation}
Using
\begin{equation}
    \hat{m}_x=\alpha_x (\hat{\chi}_x^{(0)})^{\nicefrac{1}{2}},
    ~~
    \hat{m}_y=\alpha_y (\hat{\chi})^{\nicefrac{1}{2}},
    \label{eq:m-from-alpha-clean}
\end{equation}
together with Eq.~\eqref{eq:var-rel-ordered} and $q=q_y$,
\begin{equation}
    (\hat{\chi}_x^{(0)})^{\nicefrac{1}{2}}=\sigma^{\nicefrac{1}{2}}\gamma^{-1/4}q(\alpha_y)^{\nicefrac{1}{2}},
    ~~
    \hat{\chi}^{\nicefrac{1}{2}}=\sigma^{\nicefrac{1}{2}}\gamma^{1/4}q(\alpha_x)^{\nicefrac{1}{2}}.
    \label{eq:variance-sqrt-clean}
\end{equation}
Substituting into Eq.~\eqref{eq:chem-balance-clean} gives
\begin{equation}
    \gamma^{\nicefrac{1}{2}}\alpha_x (q(\alpha_y))^{\nicefrac{1}{2}}+\alpha_y (q(\alpha_x))^{\nicefrac{1}{2}}=0,
\end{equation}

\subsection{Value Density}
\label{app:oztl-value}

We derive the value density and make explicit the cancellation of the overlap-polynomial remnants.

After inserting Eq.~\eqref{eq:collision-ansatz} and Eqs.~\eqref{eq:scaling-x-conj}--\eqref{eq:scaling-y-conj} into the finite-temperature RS/1RSB saddle functional of Claim~\ref{thm:finiteT},
the $\Theta(\beta_{\min})$ contribution can be written as
\begin{equation}
    \frac{E_0}{L}
    =
    -\frac{1}{2}\hat{\chi}_x^{(0)}\chi_x
    +\frac{1}{2}\hat{\chi}(\chi-\Delta)
    +\int Dz V_x(z)
    -\int Dz V_y(z)
    +\mac{R},
    \label{eq:E0-decomp}
\end{equation}
where $V_x$ and $V_y$ are defined as
\begin{equation}
    V_x(z)=\sup_{x\ge 0}\ab(h_x(z)x-\frac{\hat{Q}_x}{2}x^2),
    ~~
    V_y(z)=\sup_{y\ge 0}\ab(a(z)y-\frac{\hat{Q}_y+\hat{\Delta}}{2}y^2),
    \label{eq:VxVy-def}
\end{equation}
and $\mac{R}$ collects the $\mac{O}(\beta_{\min})$ remnants of the overlap polynomials:
\begin{equation}
    \mac{R}
    =
    \frac{1}{2}\ab(\hat{Q}_x q_x-(\hat{Q}_y+\hat{\Delta})q)
    +\frac{\sigma}{2}\ab(\gamma^{\nicefrac{1}{2}}\chi_x q+\gamma^{-\nicefrac{1}{2}}q_x(\Delta-\chi)).
\label{eq:R-remainder}
\end{equation}
Since $V_x$ and $V_y$ are concave quadratic maximizations, $V_x(z)=\nicefrac{h_x(z)x^\star(z)}{2}$ and $V_y(z)=\nicefrac{a(z)y^\star(z)}{2}$.
Using $\int Dz x^\star(z)=1$, $\int Dz y^\star(z)=1$, and the truncated-moment identities from Lemma~\ref{lem:trunc-gauss-AB}, together with Eq.~\eqref{eq:chi-x-clean} and Eq.~\eqref{eq:chi-minus-Delta-clean}, we obtain
\begin{equation}
    \int Dz V_x(z)=\frac{1}{2}\ab(\hat{m}_x+\hat{\chi}_x^{(0)}\chi_x),
    ~~
    \int Dz V_y(z)=\frac{1}{2}\ab(\hat{m}_y+\hat{\chi}(\chi-\Delta)).
    \label{eq:site-values}
\end{equation}
Substituting Eq.~\eqref{eq:site-values} into Eq.~\eqref{eq:E0-decomp} yields
\begin{equation}
    \frac{E_0}{L}=\frac{1}{2}\ab(\hat{m}_x-\hat{m}_y )+\mac{R}.
\end{equation}
Finally, $\mac{R}=0$ at the ordered saddle. Indeed, using $\hat{Q}_y=0$, $\hat{\Delta}=\sigma\gamma^{\nicefrac{1}{2}}\chi_x$, and
$\hat{Q}_x=\nicefrac{\sigma(\chi-\Delta)}{\gamma^{\nicefrac{1}{2}}}$ alongside $\Delta-\chi=-(\chi-\Delta)$,
the two brackets in Eq.~\eqref{eq:R-remainder} cancel exactly; hence, $\mac{R}=0$ and
\begin{equation}
    \frac{E_0}{L}=\frac{1}{2}\ab(\hat{m}_x-\hat{m}_y).
    \label{eq:value-half-diff-clean}
\end{equation}
Using Eq.~\eqref{eq:m-from-alpha-clean} and Eq.~\eqref{eq:variance-sqrt-clean},
\begin{equation}
    \hat{m}_x=\sigma^{\nicefrac{1}{2}}\gamma^{-1/4}\alpha_x(q(\alpha_y))^{\nicefrac{1}{2}},
    ~~
    \hat{m}_y=\sigma^{\nicefrac{1}{2}}\gamma^{1/4}\alpha_y(q(\alpha_x))^{\nicefrac{1}{2}}.
    \label{eq:m-explicit-clean}
\end{equation}
Substituting into Eq.~\eqref{eq:value-half-diff-clean} yields
\begin{equation}
    \frac{E_0}{L}
    =
    \frac{\sigma^{\nicefrac{1}{2}}}{2}\ab(
    \gamma^{-1/4}\alpha_x(q(\alpha_y))^{\nicefrac{1}{2}}
    - \gamma^{1/4}\alpha_y(q(\alpha_x))^{\nicefrac{1}{2}}),
    \label{eq:f-gamma-clean}
\end{equation}
which is Eq.~\eqref{eq:f-gamma} in the main text.

\section{Perturbation Derivation}
\label{app:perturbation-proofs}

This appendix proves the two perturbative statements used in Section~\ref{sec:perturbative-summaries-main}:
(i) the weak-disorder expansion in $\sigma$ at finite temperature and (ii) the small-imbalance expansion in $\gamma$ around $\gamma=1$
in the ordered $T=0$ theory. We quote standard facts about scaled simplices as presented in Lemmas~\ref{lem:simplex-delta-trunc} and \ref{lem:scaled-simplex-volume},
Gaussian identities as in Lemmas~\ref{lem:gauss-ibp} and \ref{lem:hs}, and the truncated Gaussian functions $A(\alpha),B(\alpha)$, as stated in Lemma~\ref{lem:trunc-gauss-AB}.

\subsection{Small Randomness Limit}
\label{app:sigma-expansion}

\begin{proof}
Fix $\gamma>0$ and $\beta_{\max},\beta_{\min}>0$, set $\beta=\beta_{\max}$, and write $k=-\nicefrac{\beta_{\min}}{\beta_{\max}}<0$.

\paragraph{Entropy Baseline at $\sigma=0$.}
At $\sigma=0$ the payoff vanishes, hence $Z_y(\B{x};\B{C})=\int_{\mac{Y}_M}d\B{y}1=\mathrm{Vol}(\mac{Y}_M)$ is constant in $\B{x}$ and $\B{C}$.
Therefore
\begin{equation}
    Z(\B{C})=\int_{\mac{X}_N}d\B{x}\mathrm{Vol}(\mac{Y}_M)^k
    =
    \mathrm{Vol}(\mac{X}_N)\mathrm{Vol}(\mac{Y}_M)^k,
\end{equation}
and
\begin{equation}
    \Phi(\beta_{\max},\beta_{\min};\B{C})
    =
    -\frac{1}{\beta_{\min}}\ab(\log\mathrm{Vol}(\mac{X}_N)+k\log\mathrm{Vol}(\mac{Y}_M)).
\end{equation}
Dividing by $L=\sqrt{NM}$ and using Lemma~\ref{lem:simplex-volume-asympt} gives the entropy-only limit
\begin{equation}
    \nu(0)
    =
    -\frac{\gamma^{\nicefrac{1}{2}}}{\beta_{\min}}+\frac{1}{\beta_{\max}\gamma^{\nicefrac{1}{2}}},
\end{equation}
which is $v_{\mathrm{ent}}$ in the main text.

\paragraph{Envelope Identity for $\partial_\sigma v$.}
From Claim~\ref{thm:finiteT}, $\nu(\sigma)=-(1/\beta_{\min})\mathrm{extr}_{\theta}g(\theta;\sigma)$.
Lemma~\ref{lem:envelope} yields
\begin{equation}
    \frac{\partial v}{\partial\sigma}
    =
    -\frac{1}{\beta_{\min}}\frac{\partial g}{\partial\sigma}\Big|_{\theta=\theta^\star(\sigma)}.
\end{equation}
In $g$, $\sigma$ appears only in the explicit overlap polynomial, hence at the saddle
\begin{equation}
    \frac{\partial v}{\partial\sigma}
    =
    -\frac{\beta^2}{2\beta_{\min}}
    \ab(kQ_xQ_y+k(k-1)Q_xq_1-k^2q_xq_0).
    \label{eq:dvdsigma-app}
\end{equation}

\paragraph{The $\mac{O}(\sigma)$ Coefficient.}
At $\sigma=0$, the one-site measures are uniform on $\mac{X}_N$ and $\mac{Y}_M$, so by Lemma~\ref{lem:uniform-simplex-overlaps},
\begin{equation}
Q_x=2+o(1),
~~
Q_y=2+o(1),
~~
q_x=q_0=q_1=1+o(1),
~~
L\to\infty.
\end{equation}
Substituting into the bracket of Eq.~\eqref{eq:dvdsigma-app} gives
\begin{equation}
kQ_xQ_y+k(k-1)Q_xq_1-k^2q_xq_0
=
k(k+2)+o(1),
\end{equation}
and therefore
\begin{equation}
\left.\frac{\partial v}{\partial\sigma}\right|_{\sigma=0}
=
-\frac{\beta^2}{2\beta_{\min}}k(k+2)
=
\beta_{\max}-\frac{\beta_{\min}}{2}.
\end{equation}
This is the claimed $\mac{O}(\sigma)$ coefficient and is $\gamma$-independent.

\paragraph{The $\mac{O}(\sigma^{2})$ Coefficient.}
Write
\begin{equation}
P(\sigma)=kQ_xQ_y+k(k-1)Q_xq_1-k^2q_xq_0,
~~
P(\sigma)=P(0)+\sigma P^{(1)}+\mac{O}(\sigma^{2}),
\label{eq:P-exp-app}
\end{equation}
so integrating Eq.~\eqref{eq:dvdsigma-app} gives
\begin{equation}
\nu(\sigma)
=
\nu(0)
-\frac{\beta^2}{2\beta_{\min}}
\ab(\sigma P(0)+\frac{\sigma^2}{2}P^{(1)})
+\mac{O}(\sigma^3).
\label{eq:v-int-app}
\end{equation}
Thus it remains to compute $P^{(1)}$.

The first-order responses of overlaps follow from linearizing the saddle equations at $\sigma=0$:
the conjugate relations fix the leading scaling of $(\hat\chi_x,\hat\chi_0,\hat\chi_1,\hat{Q}_x)$ as $\mac{O}(\sigma)$, and a Taylor expansion of the one-site moments
(with $m_x,m_y$ adjusted to preserve $\int Dz\langle x\rangle_z=1$ and $\int Dz\langle\langle y\rangle\rangle_{\eta|z}=1$)
yields linear relations between $(Q_x,q_x,Q_y,q_0,q_1)$ and the conjugates.
Carrying out this standard linear response calculation gives the slopes
\begin{equation}
Q_x^{(1)}=\frac{2\beta^2}{\gamma^{\nicefrac{1}{2}}}k(k+1),
~~
q_x^{(1)}=\frac{\beta^2}{\gamma^{\nicefrac{1}{2}}}k^2,
~~
Q_y^{(1)}=4\gamma^{\nicefrac{1}{2}}\beta^2,
~~
q_0^{(1)}=\gamma^{\nicefrac{1}{2}}\beta^2,
~~
q_1^{(1)}=2\gamma^{\nicefrac{1}{2}}\beta^2,
\label{eq:overlap-slopes-app}
\end{equation}
where $Q_x^{(1)}=\left. d Q_x/d\sigma \right|_{\sigma=0}$, etc.

Differentiating $P(\sigma)$ at $\sigma=0$ and using $(Q_x,Q_y,q_x,q_0,q_1)=(2,2,1,1,1)$ gives
\begin{equation}
P^{(1)}
=
k(k+1)Q_x^{(1)}+2kQ_y^{(1)}+2k(k-1)q_1^{(1)}-k^2\ab(q_x^{(1)}+q_0^{(1)}).
\label{eq:P1-app}
\end{equation}
Substituting Eq.~\eqref{eq:overlap-slopes-app} into Eq.~\eqref{eq:P1-app} yields
\begin{equation}
P^{(1)}
=
\beta^2\ab[
\frac{k^2(k^2+4k+2)}{\gamma^{\nicefrac{1}{2}}}
+
\gamma^{\nicefrac{1}{2}}(4k+3k^2)
].
\label{eq:P1-final-app}
\end{equation}
Combining Eq.~\eqref{eq:v-int-app}, $P(0)=k(k+2)$, and Eq.~\eqref{eq:P1-final-app} gives the expansion
$\nu(\sigma)=\nu(0)+v_1\sigma+v_2\sigma^2+\mac{O}(\sigma^3)$ stated in Proposition~\ref{prop:sigma-expansion-main}
(after eliminating $k=-\nicefrac{\beta_{\min}}{\beta_{\max}}$).
\end{proof}

\subsection{Imbalance Ratio Expansion}
\label{app:gamma-expansion}

\begin{proof}
Let $\gamma=1+\varepsilon$ with $|\varepsilon|\ll 1$.
The ordered $T=0$ RS saddle is characterized by the two equations
\begin{equation}
\Phi(\alpha_y)=\gamma\Phi(\alpha_x),
~~
\gamma^{\nicefrac{1}{2}}\alpha_x(q(\alpha_y))^{\nicefrac{1}{2}}+\alpha_y(q(\alpha_x))^{\nicefrac{1}{2}}=0,
\label{eq:oztl-twoeq-app}
\end{equation}
together with the value formula $E_0/L=\sigma^{\nicefrac{1}{2}}f(\gamma)$, where
\begin{equation}
f(\gamma)=\frac{1}{2}\ab(
\gamma^{-1/4}\alpha_x(q(\alpha_y))^{\nicefrac{1}{2}}
-
\gamma^{1/4}\alpha_y(q(\alpha_x))^{\nicefrac{1}{2}}
).
\label{eq:f-gamma-app}
\end{equation}
We expand around the balanced point $(\gamma,\alpha_x,\alpha_y)=(1,0,0)$.

\paragraph{Taylor Data at Origin.}
We use
\begin{equation}
\Phi(a)=\frac{1}{2}+\frac{a}{\sqrt{2\pi}}+\mac{O}(a^{3}),
\end{equation}
and the Taylor data for $q$ from Lemma~\ref{lem:q-taylor}:
\begin{equation}
q(0)=\pi,
~~
q^{\prime}(0)=-\sqrt{2\pi}(\pi-2).
\label{eq:Taylor-data-app}
\end{equation}

\paragraph{Solving for $\alpha_x$ and $\alpha_y$.}
Assume expansions
\begin{equation}
\alpha_x=a_1\varepsilon+a_2\varepsilon^2+\mac{O}(\varepsilon^3),
~~
\alpha_y=b_1\varepsilon+b_2\varepsilon^2+\mac{O}(\varepsilon^3).
\label{eq:alpha-ansatz-app}
\end{equation}
Substituting Eq.~\eqref{eq:alpha-ansatz-app} into $\Phi(\alpha_y)=\gamma\Phi(\alpha_x)$ and matching coefficients yields
\begin{equation}
b_1-a_1=\sqrt{\frac{\pi}{2}},
~~
b_2-a_2-a_1=0.
\label{eq:rel-support-app}
\end{equation}
Next expand the balance equation in Eq.~\eqref{eq:oztl-twoeq-app} using
$\gamma^{\nicefrac{1}{2}}=1+\varepsilon/2+\mac{O}(\varepsilon^2)$ and
\begin{equation}
\sqrt{q(\alpha)}=\sqrt{\pi}+\frac{q^{\prime}(0)}{2\sqrt{\pi}}\alpha+\mac{O}(\alpha^2),
\end{equation}
which gives at $\mac{O}(\varepsilon)$
\begin{equation}
a_1+b_1=0.
\label{eq:rel-balance1-app}
\end{equation}
Solving Eq.~\eqref{eq:rel-support-app} and Eq.~\eqref{eq:rel-balance1-app} yields
\begin{equation}
a_1=-\sqrt{\frac{\pi}{8}},
~~
b_1=\sqrt{\frac{\pi}{8}}.
\end{equation}
Matching the $\mac{O}(\varepsilon^2)$ coefficients in the balance equation, using Eq.~\eqref{eq:Taylor-data-app} and $b_2=a_1+a_2$ from Eq.~\eqref{eq:rel-support-app},
gives
\begin{equation}
a_2=\frac{\sqrt{2\pi}}{16}(5-\pi),
~~
b_2=\frac{\sqrt{2\pi}}{16}(1-\pi),
\end{equation}
which proves the threshold expansions in Proposition~\ref{prop:gamma-expansion-main}.

Insert the threshold expansions into Eq.~\eqref{eq:f-gamma-app}.
We use
\begin{equation}
\gamma^{\pm 1/4}
=
(1+\varepsilon)^{\pm 1/4}
=
1\pm \frac{\varepsilon}{4}+\frac{\pm 1/4(\pm 1/4-1)}{2}\varepsilon^2+\mac{O}(\varepsilon^3),
\end{equation}
together with $\sqrt{q(\alpha)}=\sqrt{\pi}+(q^{\prime}(0)/(2\sqrt{\pi}))\alpha+\mac{O}(\varepsilon^2)$.
Keeping terms up to $\mac{O}(\varepsilon^2)$ yields the expansion $f(\gamma)=f(1+\varepsilon)$ stated in Proposition~\ref{prop:gamma-expansion-main},
and therefore $E_0/L$ gives the corresponding statement for the ordered $T=0$ value.
\end{proof}

\section{Mathematical preliminaries}
\label{app:math-tools}

We collect standard tools that are repeatedly used in the appendices.  The material is organized by topic:
(i) the geometry of scaled simplices, (ii) Laplace's method on compact sets, (iii) Gaussian identities, and (iv) auxiliary identities
for replica and perturbative arguments.

\subsection{Geometry of Scaled Simplices}

\begin{lemma}
    \label{lem:simplex-delta-trunc}
    For $N\in\mab{N}$ and any measurable $f:[0,\infty)^N\to\mab{R}$ for which the integrals below are well-defined,
    \begin{equation}
    \int_{\mac{X}_N}d\B{x} f(\B{x})
    =
    \int_{[0,\infty)^N}\ab(\prod_{i=1}^N d x_i)
    \delta\ab(N-\sum_{i=1}^N x_i)f(\B{x}).
    \end{equation}
    Moreover, the domain can be truncated exactly:
    \begin{equation}
    \int_{\mac{X}_N}d\B{x} f(\B{x})
    =
    \int_{[0,N]^N}\ab(\prod_{i=1}^N d x_i)
    \delta\ab(N-\sum_{i=1}^N x_i)f(\B{x}).
    \end{equation}
    The same statements hold with $\mac{X}_N$ replaced by $\mac{Y}_M$ and $N$ replaced by $M$.
\end{lemma}

\begin{proof}
The first identity represents the standard form of integration over the affine slice $\sum_i x_i=N$ within the nonnegative orthant,
as required by the Dirac delta. On the support of $\delta(N-\sum_i x_i)$, we have $\sum_i x_i=N$ and $x_i\ge 0$; hence,
$x_i\le N$ for every $i$. Therefore, the integrand vanishes outside $[0,N]^N$, making the restriction of the domain exact.
\end{proof}

\begin{lemma}
\label{lem:scaled-simplex-volume}
The volumes of the simplex are 
\begin{equation}
\mathrm{Vol}(\mac{X}_N)=\frac{N^{N-1}}{(N-1)!},
~~~
\mathrm{Vol}(\mac{Y}_M)=\frac{M^{M-1}}{(M-1)!}.
\end{equation}
\end{lemma}

\begin{proof}
Set $x_i=N u_i$ with $u_i\ge 0$ and $\sum_{i=1}^N u_i=1$. Then $\prod_i d x_i=N^N\prod_i d u_i$ and
\begin{equation}
\delta\ab(N-\sum_{i=1}^N x_i)
=
\delta\ab(N\ab(1-\sum_{i=1}^N u_i))
=
\frac{1}{N}\delta\ab(1-\sum_{i=1}^N u_i).
\end{equation}
Hence
\begin{equation}
\mathrm{Vol}(\mac{X}_N)
=
N^{N-1}\int_{[0,\infty)^N}\ab(\prod_{i=1}^N d u_i)\delta\ab(1-\sum_{i=1}^N u_i)
=
N^{N-1}\mathrm{Vol}(\Delta_{N-1}),
\end{equation}
where $\Delta_{N-1}=\{\B u\in\mab{R}^N:\ u_i\ge 0, \sum_i u_i=1\}$ is the standard simplex.
Using $\mathrm{Vol}(\Delta_{N-1})=1/(N-1)!$ demonstrates the claim. The proof for $\mac{Y}_M$ is identical.
\end{proof}

\begin{lemma}
\label{lem:simplex-volume-asympt}
Let $N,M\to\infty$ with $N/M=\gamma\in(0,\infty)$ be fixed and $L=\sqrt{NM}$. Then
\begin{equation}
\frac{1}{L}\log\mathrm{Vol}(\mac{X}_N)=\gamma^{\nicefrac{1}{2}}+o(1),
~~
\frac{1}{L}\log\mathrm{Vol}(\mac{Y}_M)=\frac{1}{\gamma^{\nicefrac{1}{2}}}+o(1).
\end{equation}
\end{lemma}

\begin{proof}
By Lemma~\ref{lem:scaled-simplex-volume}, $\log\mathrm{Vol}(\mac{X}_N)=(N-1)\log N-\log (N-1)!$.
Stirling's formula gives $\log (N-1)!=(N-1)\log (N-1)-(N-1)+\mac{O}(\log N)$; thus, $\log\mathrm{Vol}(\mac{X}_N)=N+\mac{O}(\log N)$.
Dividing by $L$ yields $(1/L)\log\mathrm{Vol}(\mac{X}_N)=N/L+o(1)=\gamma^{\nicefrac{1}{2}}+o(1)$. The statement $\mac{Y}_M$ is identical.
\end{proof}

\begin{lemma}
\label{lem:uniform-simplex-overlaps}
Let $\B{x},\B{x}^{\prime}\sim\mathrm{Unif}(\mac{X}_N)$ be independent and define
\begin{equation}
Q(\B{x})=\frac{1}{N}\sum_{i=1}^N x_i^2,
~~
q(\B{x},\B{x}^{\prime})=\frac{1}{N}\sum_{i=1}^N x_i x_i^{\prime}.
\end{equation}
Then, as $N\to\infty$,
\begin{equation}
\mab{E}Q(\B{x})=2+o(1),
~~
\mab{E}q(\B{x},\B{x}^{\prime})=1+o(1).
\end{equation}
The same statements apply to $\mac{Y}_M$ as to $M\to\infty$.
\end{lemma}

\begin{proof}
Write $x_i=Np_i$ where $\B p\sim\mathrm{Unif}(\Delta_{N-1})$; i.e., $\B p$ is Dirichlet$(1,\dots,1)$.
Then $\mab{E}p_i^2=2/(N(N+1))$ and $\mab{E}p_i=1/N$; thus, 
$\mab{E}x_i^2=N^2\mab{E}p_i^2=2N/(N+1)$ and $\mab{E}x_i=N\mab{E}p_i=1$.
Therefore, $\mab{E}Q(\B{x})=\mab{E}x_1^2=2N/(N+1)=2+o(1)$ and
$\mab{E}q(\B{x},\B{x}^{\prime})=\mab{E}x_1\mab{E}x_1^{\prime}=1$.
\end{proof}

\subsection{Laplace Principle on Compact Sets}

\begin{lemma}
\label{lem:laplace-compact}
Let $K\subset\mab{R}^d$ be compact and let $f:K\to\mab{R}$ be continuous. Then
\begin{equation}
    \lim_{\beta\to\infty}-\frac{1}{\beta}\log\int_K \exp\ab(-\beta f(x))d x
    =
    \min_{x\in K} f(x),
\end{equation}
and
\begin{equation}
    \lim_{\beta\to\infty}\frac{1}{\beta}\log\int_K \exp\ab(\beta f(x))d x
    =
    \max_{x\in K} f(x).
\end{equation}
\end{lemma}

\begin{proof}
Write $f_*=\min_K f$. For any $\varepsilon>0$, choose $x_\varepsilon\in K$ with $f(x_\varepsilon)\le f_*+\varepsilon$.
By continuity, there exists a neighborhood $U\subset K$ with $f\le f_*+2\varepsilon$ on $U$ and $\mathrm{Vol}(U)>0$.
Thus
\begin{equation}
\int_K e^{-\beta f}d x\ge \mathrm{Vol}(U)e^{-\beta(f_*+2\varepsilon)},
\end{equation}
so
\begin{equation}
-\frac{1}{\beta}\log\int_K e^{-\beta f}d x\le f_*+2\varepsilon-\frac{1}{\beta}\log\mathrm{Vol}(U).
\end{equation}
Conversely, $f\ge f_*$ implies $\int_K e^{-\beta f}d x\le \mathrm{Vol}(K)e^{-\beta f_*}$; hence
\begin{equation}
-\frac{1}{\beta}\log\int_K e^{-\beta f}d x\ge f_*-\frac{1}{\beta}\log\mathrm{Vol}(K).
\end{equation}
Letting $\beta\to\infty$ and then $\varepsilon \to 0$ yields the first claim; the second follows by applying the first to $-f$.
\end{proof}

\subsection{Gaussian Identities}

\begin{lemma}
\label{lem:gauss-ibp}
Let $Dz=(2\pi)^{-\nicefrac{1}{2}}e^{-z^2/2}d z$. If $f$ is absolutely continuous and $\int Dz|f^{\prime}(z)|<\infty$, then
\begin{equation}
\int Dz z f(z)=\int Dz f^{\prime}(z).
\end{equation}
\end{lemma}

\begin{proof}
Since $\frac{d}{d z}\ab(e^{-z^2/2})=-z e^{-z^2/2}$, integration by parts gives
\begin{equation}
\int_{\mab{R}} z f(z)\frac{e^{-z^2/2}}{\sqrt{2\pi}}d z
=
-\int_{\mab{R}} f(z)\frac{d}{d z}\ab(\frac{e^{-z^2/2}}{\sqrt{2\pi}})d z
=
\int_{\mab{R}} f^{\prime}(z)\frac{e^{-z^2/2}}{\sqrt{2\pi}}d z,
\end{equation}
where the boundary term vanishes under the stated integrability.
\end{proof}

\begin{lemma}
\label{lem:hs}
For $a\ge 0$ and $t\in\mab{R}$,
\begin{equation}
\exp\ab(\frac{a}{2}t^2)=\int Dz\exp\ab(\sqrt{a}zt).
\end{equation}
\end{lemma}

\begin{proof}
This is the moment generating function of a standard Gaussian.
\end{proof}

\begin{lemma}
\label{lem:trunc-gauss-AB}
Let $\varphi(z)=(2\pi)^{-\nicefrac{1}{2}}e^{-z^2/2}$, $\Phi(\alpha)=\int_{-\infty}^{\alpha}\varphi(z)d z$, and $Z\sim\mac{N}(0,1)$.
Define
\begin{equation}
A(\alpha)=\mab{E}\big[\pos{Z+\alpha}\big],
~~
B(\alpha)=\mab{E}\big[\pos{Z+\alpha}^2\big].
\end{equation}
Then
\begin{equation}
A(\alpha)=\alpha\Phi(\alpha)+\varphi(\alpha),
~~
B(\alpha)=(\alpha^2+1)\Phi(\alpha)+\alpha\varphi(\alpha).
\end{equation}
Moreover,
\begin{equation}
\int_{-\alpha}^{\infty} z\varphi(z)d z=\varphi(\alpha),
~~
\int_{-\alpha}^{\infty} z^2\varphi(z)d z=\Phi(\alpha)-\alpha\varphi(\alpha).
\end{equation}
\end{lemma}

\begin{proof}
Using $A(\alpha)=\int_{-\alpha}^{\infty}(z+\alpha)\varphi(z)d z$ and $\varphi^{\prime}(z)=-z\varphi(z)$ gives
$\int_{-\alpha}^{\infty} z\varphi(z)d z=\varphi(\alpha)$ and $\int_{-\alpha}^{\infty}\varphi(z)d z=\Phi(\alpha)$, hence $A(\alpha)$.
For $B(\alpha)$, expand $(z+\alpha)^2$ and use $\int_{-\alpha}^{\infty} z^2\varphi(z)d z=\Phi(\alpha)-\alpha\varphi(\alpha)$,
which follows by one integration by parts with $u=z$ and $d v=z\varphi(z)d z$.
\end{proof}

\subsection{Replica and Variational Identities}

\begin{lemma}
\label{lem:delta-fourier}
For $t\in\mab{R}$,
\begin{equation}
\delta(t)=\int_{i\mab{R}}\frac{d m}{2\pi i}\exp(m t),
\end{equation}
where the contour $i\mab{R}$ is the imaginary axis oriented upward. Equivalently,
\begin{equation}
\delta(t)=\int_{i\mab{R}}\frac{d m}{2\pi i}\exp(-m t),
\end{equation}
after the change of variables $m\mapsto -m$.
\end{lemma}

\begin{proof}
Starting from $\delta(t)=(\nicefrac{1}{2\pi})\int_{\mab{R}}d\lambda e^{i\lambda t}$ and setting $m=i\lambda$ gives $d m=id\lambda$ and the first formula.
The second follows by $m\mapsto -m$.
\end{proof}

\begin{lemma}
\label{lem:replica-log}
Let $A(z)>0$ and assume $\int Dz|\log A(z)|<\infty$. Then
\begin{equation}
\lim_{n\to 0}\frac{1}{n}\log\int DzA(z)^n=\int Dz\log A(z).
\end{equation}
\end{lemma}

\begin{proof}
Write $\int DzA(z)^n=\int Dz\exp(n\log A(z))$ and use dominated convergence to expand at small $n$.
\end{proof}

\begin{lemma}
\label{lem:envelope}
Let $\nu(\sigma)=-(1/\beta_{\min})\mathrm{extr}_{\theta} g(\theta;\sigma)$, and assume that for $\sigma$ near $0$ there exists a differentiable selection
$\theta^\star(\sigma)$ such that $\partial_{\theta}g(\theta^\star(\sigma);\sigma)=0$ and
$g(\theta^\star(\sigma);\sigma)=\mathrm{extr}_{\theta}g(\theta;\sigma)$. Then
\begin{equation}
\frac{d v}{d\sigma}
=
-\frac{1}{\beta_{\min}}\frac{\partial g}{\partial\sigma}\Big|_{\theta=\theta^\star(\sigma)}.
\end{equation}
\end{lemma}

\begin{proof}
Differentiate $\nu(\sigma)=-(1/\beta_{\min})g(\theta^\star(\sigma);\sigma)$ and use the chain rule.
The term involving $\theta^{\star\prime}(\sigma)$ vanishes by stationarity.
\end{proof}

\begin{lemma}
\label{lem:q-taylor}
Let $A(\alpha)$ and $B(\alpha)$ be as in Lemma~\ref{lem:trunc-gauss-AB} and define $q(\alpha)=B(\alpha)/A(\alpha)^2$.
Then
\begin{equation}
q(0)=\pi,
~~
q^{\prime}(0)=-\sqrt{2\pi}(\pi-2).
\end{equation}
\end{lemma}

\begin{proof}
From Lemma~\ref{lem:trunc-gauss-AB}, $A(0)=\varphi(0)=1/\sqrt{2\pi}$, and $B(0)=\Phi(0)=\nicefrac{1}{2}$, it follows that $q(0)=\pi$.
Differentiating the closed forms yields $A^{\prime}(\alpha)=\Phi(\alpha)$ and $B^{\prime}(\alpha)=2\alpha\Phi(\alpha)+2\varphi(\alpha)$.
Evaluating at $\alpha=0$ and applying $q^{\prime}=B^{\prime}/A^2-2BA^{\prime}/A^3$ yields the stated value.
\end{proof}

\section{Details on Numerical Verification}
\label{app:numerics}

\subsection{Finite-Temperature Saddle-Point Equations}

At finite temperature, the RS/1RSB saddle is characterized by seven primary unknowns: minimizer overlaps $(Q_x,q_x)$,
maximizer overlaps $(Q_y,q_1,q_0)$, and Lagrange multipliers $(m_x,m_y)$ that enforce the simplex constraints.
The remaining conjugate variables stem from stationarity:
\begin{align}
  \hat{\chi}_x &= \frac{\sigma\beta_{\max}^2}{\gamma^{\nicefrac{1}{2}}}k^2 q_0,
  \qquad
  \hat{\chi}_0 = \gamma^{\nicefrac{1}{2}}\sigma\beta_{\max}^2q_x,
  \qquad
  \hat{\chi}_1 = \gamma^{\nicefrac{1}{2}}\sigma\beta_{\max}^2(Q_x-q_x), \\
  \hat{Q}_x &= \hat{\chi}_x - \frac{\sigma\beta_{\max}^2}{\gamma^{\nicefrac{1}{2}}}\bigl[kQ_y+k(k-1)q_1\bigr].
\end{align}

The induced single-site measure for the minimizer is a truncated Gaussian on $[0,x_{\max}]$:
\begin{equation}
  p(x\mid z)\ \propto\ \exp\Bigl(-\tfrac{\hat{Q}_x}{2}x^2+(m_x+\sqrt{\hat{\chi}_x}z)x\Bigr),
  \qquad z\sim\mac{N}(0,1),
\end{equation}
and the maximizer has a truncated exponential family on $[0,y_{\max}]$:
\begin{equation}
  p(y\mid z,\eta)\ \propto\ \exp\bigl((m_y+\sqrt{\hat{\chi}_0}z+\sqrt{\hat{\chi}_1}\eta)y\bigr),
  \qquad z,\eta\sim\mac{N}(0,1),
\end{equation}
with the 1RSB reweighting over $\eta$ proportional to $[Z_y(z,\eta)]^k$, where
$Z_y(z,\eta)=\int_0^{y_{\max}}\exp(h(z,\eta)y)dy$ and $h(z,\eta)=m_y+\sqrt{\hat{\chi}_0}z+\sqrt{\hat{\chi}_1}\eta$.

The self-consistency conditions are
\begin{equation}
  Q_x = \langle x^2\rangle,\;
  q_x = \bigl\langle \langle x\rangle_z^2 \bigr\rangle_z,\;
  Q_y = \langle y^2\rangle,\;
  q_1 = \bigl\langle \langle y\rangle_{z,\eta}^2 \bigr\rangle_{z,\eta},\;
  q_0 = \Bigl\langle \bigl(\langle \langle y\rangle_{z,\eta}\rangle_{\eta|z}\bigr)^2 \Bigr\rangle_z,
\end{equation}
together with the constraints $\langle x\rangle=1$ and $\langle y\rangle=1$.
All one-site expectations are evaluated using Gauss--Hermite quadrature in $(z,\eta)$ and closed-form expressions for the truncated integrals whenever possible.
In writing $\B{\Theta}=(Q_x,q_x,Q_y,q_1,q_0)$ and $G(\B{\Theta})$ for the moment map, we solve $\B{\Theta}=G(\B{\Theta})$ using damped fixed-point iteration
$\B{\Theta}^{(t+1)}=(1-d)\B{\Theta}^{(t)}+d G(\B{\Theta}^{(t)})$ with $d\in(0,1)$.
At each iteration, $(m_x,m_y)$ are updated by bisection to satisfy the mean constraints.
We monitor convergence using the residual $\|G(\B{\Theta})-\B{\Theta}\|$.
When sweeping over $\gamma$, we warm-start the solver in both increasing and decreasing directions; when multiple fixed points are found, we retain the one yielding the larger saddle value.

\subsection{Finite-Size Simulation: Zero Temperature via Linear Programming}
For a given realization of $C$, the min-max value $t(C)$ is computed exactly by the primal linear program
\begin{equation}
    \min_{\B{p}\in \Delta_N, u} u,~~~\mathrm{s.t.}~~\B{C}^\top \B{p} \le u \1.    
\end{equation}
We compare the finite-size scaled estimate $f_L=(NM)^{1/4}t(C)$ to the theoretical prediction.

\subsection{Finite-Size Simulation: Finite Temperature via Annealed Importance Sampling}
When $k<0$, the importance weights $[Z_y(\B{x};\B{C})]^k$ in
Eq.~\eqref{eq:ttbs-outer-def} can exhibit heavy-tailed behavior, resulting in naive uniform Monte Carlo over $\mac{X}_N$ being high variance.
To stabilize the estimator, we also employ annealed importance sampling (AIS) \citep{neal2001annealed}.
AIS introduces a series of bridging distributions on $\mac{X}_N$:
\begin{equation}
    \pi_t(\B{x})
    \propto
    p_{0}(\B{x})\exp\ab(-\beta_t \phi_{\beta_{\max}}(\B{x};\B{C})),~~
    0=\beta_0<\beta_1<\cdots<\beta_T=\beta_{\min},
    \label{eq:ais-bridge}
\end{equation}
which interpolates between the prior $p_{0}(\B{x})$ and the target outer distribution
$p_{\beta_{\min}}(\B{x};\B{C})\propto p_{0}(\B{x})\exp(-\beta_{\min}\phi_{\beta_{\max}}(\B{x};\B{C}))$.
Using $\phi_{\beta_{\max}}(\B{x};\B{C})=(\nicefrac{1}{\beta_{\max}})\log Z_y(\B{x};\B{C})$, the incremental AIS weight can be expressed in terms of
$\log Z_y(\B{x};\B{C})$.
For a single AIS run producing a trajectory $\B{x}_0,\ldots,\B{x}_{T-1}$,
the accumulated log-weight is
\begin{equation}
    \log W
    =
    -\sum_{t=1}^{T}(\beta_t-\beta_{t-1})\phi_{\beta_{\max}}(\B{x}_{t-1};\B{C})
    =
    \frac{k}{\beta_{\min}}\sum_{t=1}^{T}(\beta_t-\beta_{t-1})\log Z_y(\B{x}_{t-1};\B{C}).
    \label{eq:ais-logweight}
\end{equation}
Each stage employs an MCMC transition kernel targeting $\pi_t$.
We implement a Metropolis–Hastings update with a symmetric pairwise-exchange proposal that preserves the simplex constraint:
two coordinates $(x_i,x_j)$ are perturbed by $\pm\delta$ with $\delta\sim\mathrm{Uniform}(-h,h)$, followed by rejection if the proposal exits $\mac{X}_N$. Given $S$ independent AIS runs with weights $\{W_s\}_{s=1}^{S}$, we estimate the outer partition function in Eq.~\eqref{eq:ttbs-outer-def} by
\begin{equation}
    Z(\B{C})
    =
    \int_{\mac{X}_N}p_0(\B{x})\exp\ab(-\beta_{\min}\phi_{\beta_{\max}}(\B{x};\B{C})) d\B{x}
     \approx
    \frac{1}{S}\sum_{s=1}^{S} W_s,
\end{equation}
and report the corresponding free-energy estimate
\begin{equation}
    \hat{F}(\beta_{\max},\beta_{\min};\B{C})
    =
    -\frac{1}{\beta_{\min}}\log \hat{Z}(\B{C}),~~
    \hat{Z}(\B{C})=\frac{1}{S}\sum_{s=1}^{S} W_s .
    \label{eq:ais-Fhat}
\end{equation}

\section{Additional Numerical Comparison}\label{sec:additional-comparison}

\subsection{Additional Verification: Finite-Temperature Free Energy at Finite Size}
\label{app:finiteT-free-energy}

In addition to the equilibrium-value comparisons mentioned in the main text, we further verify that our finite-temperature free energy predictions remain accurate for finite game sizes.
 At finite temperatures, the relevant quantity is the free energy
$\Phi(\beta_{\max},\beta_{\min};\B{C})=-(1/\beta_{\min})\log Z(\B{C})$, and our theory predicts its normalized typical value
$\nu(\beta_{\max},\beta_{\min})=\lim_{L\to\infty}(1/L)\mab E_{\B{C}}[\Phi(\beta_{\max},\beta_{\min};\B{C})]$.
To investigate this prediction in finite instances, we estimate $\Phi(\beta_{\max},\beta_{\min};\B{C})$ for each sampled payoff matrix
using the AIS procedure described in Appendix~\ref{app:numerics}, reporting the normalized estimate
$\hat{v} = \nicefrac{\hat{F}}{L}$, with $\hat{F}(\beta_{\max},\beta_{\min};\B{C})$ defined in Eq.~\eqref{eq:ais-Fhat} and $L=\sqrt{NM}$
We employ the following hyperparameters and numerical pipeline outlined in Appendix~\ref{app:numerics};
we fix the maximizer's strategy count to $M=80$ and set $N=\gamma M$ with $\gamma\in\{0.8,1.0,1.2\}$; the payoff scale is $\sigma=1.0$.
We consider $\beta_{\max}\in\{0.1,0.5,1.0\}$ and sweep the temperature ratio
$k=-\nicefrac{\beta_{\min}}{\beta_{\max}}\in\{-0.75,-0.80,\ldots,-1.25\}$.
For AIS, we conduct 10 independent trials with 3,000 parallel chains, 1,000 intermediate temperatures, and 5 MCMC steps per temperature.

Figure~\ref{fig:finiteT-free-energy-sweep-k} shows that the replica prediction for $\nu(\beta_{\max},\beta_{\min})$
closely matches the empirical finite-size estimates across all tested settings.
Importantly, this agreement persists not only in the zero-temperature regime but also at genuinely finite temperatures,
indicating that our nested replica theory captures the finite temperature free-energy landscape rather than merely its
ordered zero-temperature limit.
Overall, these results reinforce the practical message of the paper: even for moderate sizes, $M=80$, the asymptotic theory is already quantitatively predictive for both bounded rationality and near-rationality regimes.

\begin{figure}[tb]
    \centering
    \includegraphics[width=\textwidth]{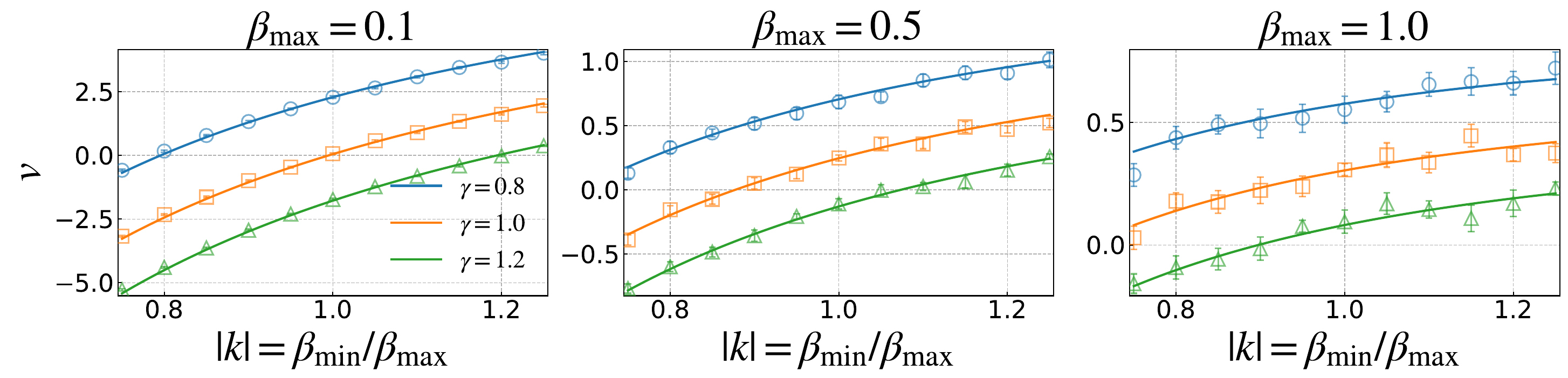}
    \caption{Finite-temperature normalized free energy $\nu$ as a function of $|k|=\nicefrac{\beta_{\min}}{\beta_{\max}}$: markers denote finite-size estimates and solid curves are the replica predictions, shown for $\beta_{\max}\in\{0.1,0.5,1.0\}$ and $\gamma\in\{0.8,1.0,1.2\}$.}
    \label{fig:finiteT-free-energy-sweep-k}
\end{figure}

\section{Relations to Existing Relaxed Game}
\label{app:bounded_rationality}

This appendix situates our two-temperature objective within the broader landscape of bounded rationality models in game theory by establishing three precise connections to standard formulations: (i) in single-agent decision making, temperature-smoothed maxima and minima are \emph{exactly} characterized by the Gibbs variational principle, which maximizes expected utility subject to an explicit KL information cost; (ii) in two-stage extensive-form settings, our nested two-temperature construction is equivalent to a KL-regularized hierarchical game, where the maximizer conditionally best responds to the minimizer’s action, and each player’s temperature independently specifies its own information/precision budget, reducing for finite action sets \emph{exactly} to node-wise logit choice, and thus to the logit agent-quantal-response equilibrium (AQRE); and (iii) by contrast, the commonly used simultaneous entropy and KL-regularized saddle formulation yields the well-known logit best-response fixed point in normal-form games, rather than the nested structure induced by our objective. We conclude by emphasizing that several significant departures from perfect rationality are not captured by temperature and KL regularization alone, including depth-limited reasoning, misspecified beliefs, and explicit computational constraints.

\subsection{Soft Max and Min as KL-regularized Optimization}
\label{app:gibbs-variational}

\begin{proposition}[Gibbs Variational Formula]
\label{prop:gibbs-dv}
Let $q_0$ be a prior density on $\mac{Y}$ and let $f:\mac{Y}\to\mab R$ satisfy
$\int_{\mac{Y}} q_0(y)e^{\beta f(y)}dy<\infty$.
Then for any $\beta>0$,
\begin{equation}
    \frac{1}{\beta}\log \int_{\mac{Y}} q_0(y)e^{\beta f(y)}dy
    =
    \sup_{q\in\Delta(\mac{Y})}
    \ab\{
    \int_{\mac{Y}} q(y) f(y)dy
    -\frac{1}{\beta}\mathrm{KL}(q\|q_0)\},
\label{eq:gibbs-var-max}
\end{equation}
where $\Delta(\mac{Y})=\{q\ge 0:\int_{\mac{Y}} q(y) dy=1\}$ and
$\mathrm{KL}(q\|q_0)=\int q(y)\log(\nicefrac{q(y)}{q_0(y)})dy$.
The supremum is attained at the Gibbs density
\begin{equation}
    q^\star_\beta(y)
    =
    \frac{q_0(y)e^{\beta f(y)}}{\int_{\mac{Y}} q_0(s)e^{\beta f(s)}ds}.
\label{eq:gibbs-opt}
\end{equation}
Similarly, for any $g:\mac{X}\to\mab R$ with $\int_{\mac{X}} p_0(x)e^{-\beta g(x)}dx<\infty$,
\begin{equation}
    -\frac{1}{\beta}\log \int_{\mac{X}} p_0(x)e^{-\beta g(x)}dx
    =
    \inf_{p\in\Delta(\mac{X})}
    \ab\{
    \int_{\mac{X}} p(x) g(x)dx
    +\frac{1}{\beta}\mathrm{KL}(p\|p_0)
    \}.
    \label{eq:gibbs-var-min}
\end{equation}
\end{proposition}

\begin{proof}
Fix any $q\in\Delta(\mac{Y})$ with $q(y)>0$ where $q_0(y)>0$ and define
\begin{equation}
    \int q_0(y)e^{\beta f(y)}dy
    =
    \int q(y)\exp \ab(\beta f(y)+\log\frac{q_0(y)}{q(y)})dy.    
\end{equation}
Applying Jensen's inequality to $\log \int q(\cdot)\exp(\cdot)$ gives
\begin{equation}
    \log \ab(\int q_0(y) e^{\beta f(y)} dy) 
    \ge
    \int q(y)\ab(\beta f(y)+\log\frac{q_0(y)}{q(y)})dy
    =
    \beta\int q(y) f(y) dy - \mathrm{KL}(q\|q_0).    
\end{equation}
Dividing by $\beta$ and taking the supremum over $q$ yields Eq.~\eqref{eq:gibbs-var-max}.
Equality holds iff the exponent is constant on the support, i.e., $q\propto q_0 e^{\beta f}$, which gives Eq.~\eqref{eq:gibbs-opt}.
The min identity follows by applying the max identity to $-g$.
\end{proof}

Proposition~\ref{prop:gibbs-dv} formalizes the standard free-energy view of bounded rationality
\citep{OrtegaBraun2013ThermodynamicsDecision}: temperature $\beta^{-1}$ controls a trade-off between expected payoff and an explicit KL information cost. As $\beta\to\infty$, the KL penalty vanishes and the soft operator approaches the hard max and min; as $\beta\to 0$, the prior dominates and the optimizer approaches the prior distribution.

We connect the nested two-temperature construction to a two-stage game as defined in Definition \ref{def:ttbs}.

\begin{proposition}[Two-temperature objective as a KL-regularized hierarchical game]
\label{prop:two-temp-kl}
With priors $(p_0,q_0)$, the two-temperature quantities in Definition \ref{def:ttbs} have exact variational forms
\begin{align}
    \phi_{\beta_{\max}}(\B{x})
    &=
    \sup_{q(\cdot|\B{x})\in\Delta(\mac{Y})}
    \ab\{
    \int_{\mac{Y}} q(\B{y}|\B{x}) E(\B{x},\B{y};\B{C}) d\B{y}
    -\frac{1}{\beta_{\max}}\mathrm{KL}\ab(q(\cdot|\B{x})\|q_0)\},
    \label{eq:inner-br}
    \\
    \Phi(\beta_{\max},\beta_{\min})
    &=
    \inf_{p\in\Delta(\mac{X})}
    \ab\{
    \int_{\mac{X}} p(\B{x}) \phi_{\beta_{\max}}(\B{x}) d\B{x}
    +\frac{1}{\beta_{\min}}\mathrm{KL}(p\|p_0)\}.
    \label{eq:outer-br}
\end{align}
For finite $\mac{X}$, substituting Eq.~\eqref{eq:inner-br} into Eq.~\eqref{eq:outer-br} yields
\begin{multline}
    \Phi(\beta_{\max},\beta_{\min})
    \\
    =
    \inf_{p}
    \sup_{q(\cdot|\B{x})}
    \ab\{
    \mab E_{\B{X}\sim p,\B{Y}\sim q(\cdot|\B{X})}[E(\B{X},\B{Y}; \B{C})]
    +\frac{1}{\beta_{\min}}\mathrm{KL}(p\|p_0)
    -\frac{1}{\beta_{\max}}\mab E_{\B{X}\sim p}\ab[\mathrm{KL}\ab(q(\cdot|\B{X})\|q_0)]\}.
    \label{eq:stackelberg-kl-game}
\end{multline}
\end{proposition}

\begin{proof}
Eq.~\eqref{eq:inner-br} is Proposition~\ref{prop:gibbs-dv} applied pointwise to $f(\B{y})=E(\B{x},\B{y}; \B{C})$ with prior $q_0$.
Eq.~\eqref{eq:outer-br} is Proposition~\ref{prop:gibbs-dv} applied to $g(\B{x})=\phi_{\beta_{\max}}(\B{x})$ with prior $p_0$.
The combined expression follows by substitution; for finite $\mac{X}$, the pointwise supremum can be taken separately for each $\B{x}$.
\end{proof}

Proposition~\ref{prop:two-temp-kl} explicitly indicates that the two temperatures function as independent resource parameters: $\beta_{\max}$ governs the precision of the maximizer conditional response $q(\cdot|\B{x})$, while $\beta_{\min}$ regulates the precision of the minimizer marginal choice $p$. This separation is central to our model: the resulting outcome is interpreted as bounded-rational leader and follower relaxation rather than a simultaneous-move equilibrium at finite temperature.

\subsection{Relation to AQRE and QRE}
\label{app:qre-relation}

We specialize in finite action sets and demonstrate that the two-temperature objective reduces precisely to logit choice at each decision node of a two-stage extensive-form game. Let $\mac{X}=\{1,\dots,N\}$ and $\mac{Y}=\{1,\dots,M\}$ with priors $p_0(i)$ and $q_0(j)$. Then, thermal min-max games become
\begin{multline}
    \phi_{\beta_{\max}}(i)
    =
    \frac{1}{\beta_{\max}}\log\sum_{j=1}^M q_0(j)\exp\ab(\beta_{\max}V_{ij}),\\
    \Phi(\beta_{\max},\beta_{\min})
    =
    -\frac{1}{\beta_{\min}}\log\sum_{i=1}^N p_0(i)\exp\ab(-\beta_{\min}\phi_{\beta_{\max}}(i)).
\label{eq:discrete-nested-logsumexp}
\end{multline}
The optimizers provided by Proposition~\ref{prop:gibbs-dv} are the corresponding logit policies:
\begin{equation}
    q_{\beta_{\max}}(j|i)
    =
    \frac{q_0(j)\exp(\beta_{\max}V_{ij})}{\sum_{s=1}^M q_0(s)\exp(\beta_{\max}V_{is})},
    ~~~
    p_{\beta_{\min}}(i)
    =
    \frac{p_0(i)\exp(-\beta_{\min}\phi_{\beta_{\max}}(i))}{\sum_{r=1}^N p_0(r)\exp(-\beta_{\min}\phi_{\beta_{\max}}(r))}.
    \label{eq:discrete-logit-policies}
\end{equation}

\begin{proposition}
\label{prop:aqre-equivalence}
Consider the two-stage extensive-form zero-sum game in which the minimizer first chooses $i\in\mac{X}$, after which the maximizer observes $i$ and selects $j\in\mac{Y}$, yielding a payoff of $V_{ij}$ for the maximizer and $-V_{ij}$ for the minimizer. Then the logit agent-quantal-response equilibrium (AQRE) with node precisions $(\beta_{\min},\beta_{\max})$
\citep{McKelveyPalfrey1998AQRE}
precisely induces the behavioral strategies Eq.~\eqref{eq:discrete-logit-policies} and the continuation-value recursion Eq.~\eqref{eq:discrete-nested-logsumexp}.
\end{proposition}

\begin{proof}
By the definition of AQRE, each information set applies logit choice to its continuation payoffs \citep{McKelveyPalfrey1998AQRE}; the follower node produces $q_{\beta_{\max}}(\cdot|i)$ and the continuation value $\phi_{\beta_{\max}}(i)$, while the leader node applies logit to the leader payoff $\phi_{\beta_{\max}}(i)$, yielding $p_{\beta_{\min}}$.
\end{proof}

This correspondence clarifies the relation to normal-form quantal response equilibrium (QRE): in simultaneous-move normal-form games, logit QRE is defined as a fixed point of quantal response maps \citep{McKelveyPalfrey1995QRE}, while our nested objective corresponds to the sequential extensive-form counterpart of the same logit and KL principle. 

A widely used simultaneous relaxation instead solves an entropy- and KL-regularized saddle problem 
for $\B{A}\in\mab R^{N\times M}$ and priors $\B{p}_0\in\Delta_N$, $\B{q}_0\in\Delta_M$,
\begin{equation}
    \min_{\B{p}\in\Delta_N}\max_{\B{q}\in\Delta_M}
    \ab\{
    \B{p}^\top \B{A} \B{q}
    +\frac{1}{\beta_{\min}}\mathrm{KL}(p\|p_0)
    -\frac{1}{\beta_{\max}}\mathrm{KL}(q\|q_0)\}.
\label{eq:entropic-saddle}
\end{equation}
With uniform priors, Eq.~\eqref{eq:entropic-saddle} is equivalent, up to constants, to adding $-(1/\beta_{\min})H(\B{p})$ and $(1/\beta_{\max})H(\B{q})$, so both players are encouraged to randomize. The first-order optimality conditions yield smooth and logit best responses, resulting in a logit fixed-point characterization.

\begin{proposition}
\label{prop:entropic-logit-fixedpoint}
Assume $(\B{p}^\star,\B{q}^\star)$ solves Eq.~\eqref{eq:entropic-saddle} and has full support.
Then it satisfies
\begin{equation}
    p^\star_i
    \propto
    p_0(i)\exp\ab(-\beta_{\min}(\B{A}\B{q}^\star)_i),~~~
    q^\star_j
    \propto
    q_0(j)\exp\ab(\beta_{\max}(\B{A}^\top \B{p}^\star)_j),
\end{equation}
which coincide with the logit quantal-response fixed point conditions for zero-sum games
\citep{McKelveyPalfrey1995QRE,GoereeHoltPalfrey2005RegularQRE}.
\end{proposition}

\begin{proof}
For fixed $\B{q}$, minimizing the strictly convex objective in Eq.~\eqref{eq:entropic-saddle} over $\B{p}\in\Delta_N$ yields the unique optimizer
$p_i \propto p_0(i)\exp(-\beta_{\min}(\B{A}\B{q})_i)$ according to the KKT conditions. Similarly, for fixed $p$, maximizing over $\B{q}$ yields
$q_j \propto q_0(j)\exp(\beta_{\max}(\B{A}^\top \B{p})_j)$. At a saddle point, both relations hold simultaneously.
\end{proof}

This simultaneous formulation is closely related to smooth learning dynamics, such as logit dynamics and noisy best-response dynamics
\citep{Blume1993LogitDynamics,HofbauerSandholm2002StochasticFP}, where temperature serves as a measure of behavioral precision or noise level, with an analogous KL-regularization interpretation.

\subsection{Limitations of Temperature and KL Relaxations}
\label{app:beyond-temperature}

The temperature and KL framework captures bounded rationality as stochasticity induced by an explicit information-processing cost; that is, it describes entropy-regularized optimization. Several influential models of imperfect rationality address different mechanisms and cannot be reduced to temperature alone. These models include depth-limited strategic reasoning, level-$k$ and cognitive hierarchy \citep{StahlWilson1995LevelK,CamererHoChong2004CH}, belief misspecification (e.g., cursed equilibrium \citep{EysterRabin2005Cursed}), and procedural or computational constraints (e.g., finite automata and algorithmic rationality \citep{Rubinstein1986FiniteAutomata,HalpernPass2015AlgorithmicRationality}). Therefore, our two-temperature relaxation should be considered a principled core model for soft-response and information-cost approaches, which can be integrated with these orthogonal behavioral axes.

\end{document}